\documentclass[]{lmcs}
\pdfoutput=1

\usepackage{lastpage}
\lmcsdoi{20}{1}{20}
\lmcsheading{}{\pageref{LastPage}}{}{}%
{Jul.~28,~2023}{Mar.~07,~2024}{}

\usepackage[utf8]{inputenc}
\usepackage{amssymb}
\usepackage{amsfonts}
\usepackage{mathtools}
\usepackage{multicol}
\usepackage{hyperref}
\usepackage{graphicx}
\usepackage{amsmath}
\usepackage{stmaryrd}
\usepackage{amsthm}

\usepackage{tikz}
\usepackage{pgfplots}
\usetikzlibrary{decorations.markings}
\usetikzlibrary{matrix,backgrounds,folding}
\usetikzlibrary{shapes.geometric, shapes.misc}
\pgfdeclarelayer{edgelayer}
\pgfdeclarelayer{nodelayer}
\pgfsetlayers{background,edgelayer,nodelayer,main}
\tikzset{every path/.style={draw=black!80, line width=0.6pt}}
\tikzstyle{every picture}=[baseline=-0.25em]
\tikzstyle{none}=[inner sep=0mm]
\tikzstyle{black dot}=[inner sep=0.7mm,minimum width=0pt,minimum height=0pt,fill=black,draw=black,shape=circle]
\tikzstyle{dot}=[black dot]
\tikzstyle{white dot}=[dot,fill=white]
\tikzstyle{box}=[rectangle,fill=white,draw=black, font=\scriptsize, inner sep=2pt]
\tikzstyle{arrow}=[decoration={markings,mark=at position 1 with
    {\arrow[scale=1.2,>=stealth]{>}}},postaction={decorate}]
\tikzstyle{box-no-outline}=[rectangle, draw=white, fill=white, inner sep=2pt]
\tikzstyle{polynomial}=[regular polygon, regular polygon sides=3, shape border rotate= 180, fill=white,draw=black,inner sep = -10pt, rounded corners, rounded corners=10pt, minimum width=25pt, font=\scriptsize]

\pgfdeclarelayer{edgelayer}
\pgfdeclarelayer{nodelayer}
\pgfsetlayers{background,edgelayer,nodelayer,main}
\tikzstyle{none}=[inner sep=0mm]
\tikzstyle{every loop}=[]

\usepackage{etoolbox}

\def\fig{}

\newcommand{\eq}[2][~~]{
#1
\underset{\substack{#2}}{=}
#1
}

\newcommand{\rewrite}[1]{
\xrightarrow[\substack{#1}]{}
}
\newcommand{\interp}[1]{\left\llbracket #1 \right\rrbracket}
\newcommand{\abs}[1]{\left\lvert #1 \right\rvert}
\newcommand{\bra}[1]{\ensuremath{\left\langle \textstyle #1 \right|}}
\newcommand{\ket}[1]{\ensuremath{\left| \textstyle #1 \right\rangle}}
\newcommand{\ketbra}[2]{\ket{#1}\!\!\bra{#2}}

\newcommand{\cat}[1]{\mathbf{#1}}
\newcommand{\atan}[1]{\operatorname{arctan}\left(#1\right)}
\newcommand{\Var}{\operatorname{Var}}

\newlength{\delimstretch}

\newcommand{\ascend}[1]{
\setlength{\delimstretch}{\heightof{\ensuremath{#1}}}
\pgfmathsetlength{\delimstretch}{max(\delimstretch,1.5ex)}
\resizebox{!}{1.1\delimstretch}{$\raisebox{0.18ex}{$\upharpoonleft$}\hspace*{-0.94ex}\lfloor$}\ensuremath{#1}\resizebox{!}{1.1\delimstretch}{$\rfloor\hspace*{-0.94ex}\raisebox{0.18ex}{$\upharpoonright$}$}
}

\newcommand{\descend}[1]{
\setlength{\delimstretch}{\heightof{\ensuremath{#1}}}
\pgfmathsetlength{\delimstretch}{max(\delimstretch,1.5ex)}
\resizebox{!}{1.1\delimstretch}{$\raisebox{-0.18ex}{$\downharpoonleft$}\hspace*{-0.94ex}\lceil$}\ensuremath{#1}\resizebox{!}{1.1\delimstretch}{$\rceil\hspace*{-0.94ex}\raisebox{-0.18ex}{$\downharpoonright$}$}
}

\allowdisplaybreaks

\begin{document}

\bibliographystyle{alphaurl}

\title[\textbf{SOP} for Dyadic Fragments of Quantum Computation]
{Rewriting and Completeness of Sum-Over-Paths in Dyadic Fragments of Quantum Computing\rsuper*}
\titlecomment{{\lsuper*}Extended version of the CSL'23 paper \cite{Vilmart2023completeness}.}
\thanks{The author acknowledges support from the PEPR integrated project EPiQ ANR-22-PETQ-0007 part of Plan France 2030, the ANR projects TaQC ANR-22-CE47-0012 and HQI ANR-22-PNCQ-0002, as well as the European project HPCQS.}

\author[R.~Vilmart]{Renaud Vilmart\lmcsorcid{0000-0002-8828-4671}}

\address{Université Paris-Saclay, CNRS, ENS Paris-Saclay, Inria, LMF, 91190, Gif-sur-Yvette, France}

\email{renaud.vilmart@inria.fr}


\maketitle

\begin{abstract}
The ``Sum-Over-Paths'' formalism is a way to symbolically manipulate linear maps that describe quantum systems, and is a tool that is used in formal verification of such systems.

We give here a new set of rewrite rules for the formalism, and show that it is complete for ``Toffoli-Hadamard'', the simplest approximately universal fragment of quantum mechanics. We show that the rewriting is terminating, but not confluent (which is expected from the universality of the fragment). We do so using the connection between Sum-over-Paths and graphical language ZH-calculus, and also show how the axiomatisation translates into the latter. We provide generalisations of the presented rewrite rules, that can prove useful when trying to reduce terms in practice, and we show how to graphically make sense of these new rules.

We show how to enrich the rewrite system to reach completeness for the dyadic fragments of quantum computation, used in particular in the Quantum Fourier Transform, and obtained by adding phase gates with dyadic multiples of $\pi$ to the Toffoli-Hadamard gate-set.

Finally, we show how to perform sums and concatenation of arbitrary terms, something which is not native in a system designed for analysing gate-based quantum computation, but necessary when considering Hamiltonian-based quantum computation.
\end{abstract}

\section{Introduction}

Sum-Over-Paths (SOP) is a formalism used to represent and manipulate quantum processes in a symbolic way, introduced in 2018 by Amy \cite{SOP}. Its first important feature is its capacity to translate from most common (gate-based) descriptions of quantum processes in polynomial time and space. The formalism hence provides a intermediary view between usual (matrix) semantics and these usual process descriptions. Its second crucial feature is that it comes equipped with a rewrite system that simplifies the term, without altering its semantics.

Despite its links \cite{MSc.Lemonnier,LvdWK} with graphical languages such as the ZH-calculus \cite{ZH} -- which will be used in the following --, it provides a different view on the quantum processes, representing them as weighted sums of Dirac kets and bras (a very familiar notation in quantum mechanics).

The formalism has seen several applications, the first of which being verification. Verification is a crucial aspect of computations in the quantum realm, where physical constraints (like no-cloning, or the fundamental probabilistic nature of quantum) make it impossible to do debugging the way we do on classical algorithms. More specifically, the SOP formalism was introduced as a solution to circuit equivalence: To check the equivalence between two circuits $\mathcal C_1$ and $\mathcal C_2$, the system represents $\mathcal C_2^\dagger \circ\mathcal C_1$ as an SOP term (where $\mathcal C_2^\dagger$ can be seen as the inverse of $\mathcal C_2$, easy to describe from it). It then tries to reduce it to the identity. If successful, this proves $\mathcal{C}_1$ and $\mathcal C_2$ represent the same unitary. Otherwise, the system searches for a witness that the term at hand does not represent the identity. As such, the system has been used in several different projects (e.g.~\cite{beaudrap2020fast,Kissinger2020reducing}) to check precisely for circuit equivalence. It was later extended to account for families of morphisms and used within the Qbricks environment \cite{chareton2020deductive,chareton2021survey} together with automated solvers to verify algorithms and routines such as quantum phase estimation, Grover's search and Shor's algorithm.

Amongst other applications of the Sum-Over-Paths, we may cite noiseless simulation of quantum processes, where the rewrite strategy is used to reduce the number of variables in the term, effectively decreasing the number of summands when expanding the term to actually compute its semantics. It is for instance one of the simulators implemented in the supercomputer Atos QLM \cite{Haidar2022qlm}.

While the initial suggestion for Sum-Over-Paths focussed on the Clifford+T fragment -- a universal fragment of quantum computing, i.e.~a restriction still capable of approximating with arbitrary precision any quantum process --, it also provided some interesting result for the Clifford fragment. It is known that the latter is not universal \cite{clifford-not-universal}, and actually efficiently simulable with a classical computer, so it is a good test for the relevance of a formalism to check how it handles them. And indeed, it was shown \cite{SOP} a ``weak'' form of confluence of the rewrite system in the Clifford fragment. More precisely, in this fragment, $\mathcal C_2^\dagger \circ\mathcal C_1$ reduces (in polynomial time) to the identity if and only if $\mathcal C_2$ and $\mathcal C_1$ represent the same unitary operator.

However, SOP terms may represent more than unitary operator, but actually any linear map. With those, it is still possible to define the above restrictions, and the rewrite system was extended in \cite{SOP-Clifford} to get confluence for the -- not necessarily unitary -- Clifford fragment. When moving to a universal fragment -- like Clifford+T -- it is expected that we cannot provide a rewrite system with all the good properties of the Clifford case: either reduction is not polynomial, or there is no confluence, or we need an infinite number of rewrites, ... The reason for this is that if we could provide such a system, deciding circuit equivalence would become polynomial, while we know that it is QMA-complete -- a quantum variant of NP-complete -- \cite{Bookatz2014qma,Janzing2005identity}. A weaker property than that of confluence we can ask for is completeness: the question here is to decide whether two equivalent terms can be turned into one another, \emph{with the assumption that rewrites can be used in both directions} (in that case, we rather speak of an equational theory, or axiomatisation, than a rewrite system).

\medskip\noindent
\textbf{Contributions.~}
In this paper, we address the problem of completeness first for arguably the simplest universal fragment of quantum computing, which is \emph{Toffoli-Hadamard}. We provide a fairly simple rewrite system that we show complete for the fragment, and also exhibit two important drawbacks: the non-confluence of the rewrite strategy and the potential explosion of the size of the morphisms during the rewrite. We then show how the rewrite strategy can be tweaked to reach completeness for every dyadic fragment -- where we allow phase gates with phase a multiple of $\frac\pi{2^k}$ for some $k$ --, a restriction that encompasses Clifford, Clifford+T and Toffoli-Hadamard, and is crucially used in the Quantum Fourier Transform, a central block for algorithms such as Shor's and Quantum Phase Estimation. This paper extends \cite{Vilmart2023completeness} which was presented at the Computer Science and Logic 2023 (CSL'23). Here we slightly simplify the rewrite system, provide further potential simplification rules, of which we give a graphical treatment, and finally show how to perform two non-trivial operations on SOP terms, namely their sum and their concatenation, both very useful when building arbitrary maps or when considering Hamiltonian-based computation, and which remain in the dyadic fragment at hand. Since the presentation of \cite{Vilmart2023completeness}, a complete rewrite system was given to Sum-Over-Paths with arbitrary phases and complex entries \cite{amy2023complete}, although in the ``unbalanced'' framework, which is noticeably different.

\medskip\noindent
\textbf{Structure of the paper.~}
We first review the Sum-Over-Paths formalism in \autoref{sec:SOP}, providing a rewrite strategy for Toffoli-Hadamard, as well as generalised rules. We then present the ZH-calculus in \autoref{sec:ZH} and the links between the two in \autoref{sec:SOP<->ZH}, and we show how to interpret the new rewrite rules of SOP graphically. We then show the completeness result for the Toffoli-Hadamard fragment in \autoref{sec:TH-completeness}. The extension to the dyadic fragments is then handled in \autoref{sec:completeness}. Finally, in \autoref{sec:sum-concat}, we show a way to ``control'' arbitrary terms allowing for their sum and their concatenation.

\section{Sums-Over-Paths}
\label{sec:SOP}

\subsection{The Morphisms}

Sums-Over-Paths \cite{SOP} are a way to symbolically describe linear maps of dimensions powers of $2$ over the complex numbers. These linear maps form a $\dagger$-compact monoidal category \cite{mac2013categories,selinger2010survey} denoted $\cat{Qubit}$ where the objects are natural numbers (this makes the category a PROP \cite{Lack-PROP,PhD.Zanasi}), where morphisms from $n$ to $m$ are linear maps $\mathbb C^{2^n}\to \mathbb C^{2^m}$, and where $(\cdot\circ\cdot)$ (resp.~$(\cdot\otimes\cdot)$) is the usual composition (resp.~tensor product) of linear maps. The category is endowed with a \emph{symmetric braiding} $\sigma_{n,m}:n+m\to m+n$, as well as a \emph{compact structure} $(\eta_n:0\to2n,\epsilon_n:2n\to0)$. Furthermore, there exists an inductive contravariant endofunctor $(\cdot)^\dagger$, that behaves properly with the symmetric braiding and the compact structure. For more information on these structures, see \cite{selinger2010survey}.

The formalism of SOP relies heavily on the Dirac notation for quantum states and operators of $\cat{Qubit}$. The two canonical states of a single qubit are denoted $\ket0$ and $\ket1$. They form a basis of $\mathbb C^2$, and can be viewed as vectors $\ket0=\begin{pmatrix}1&0\end{pmatrix}^\intercal$ and $\ket1=\begin{pmatrix}0&1\end{pmatrix}^\intercal$. A 1-qubit state is then merely a normalised linear combination of these two elements. Using $(\cdot\otimes\cdot)$, they can be used to build the basis states of larger systems, e.g.~$\ket{010}:=\ket0\otimes\ket1\otimes\ket0$ is a basis state of a 3-qubit system. Again, the state of an arbitrary $n$-qubit system is a normalised linear combination of the $2^n$ basis states. We will use extensively the following notation $\bra x$ to represent the dagger (transpose conjugate) of $\ket x$. The identity on a qubit can then be expressed in Dirac notation as $\mathbb I:=\ketbra00+\ketbra11$, where $\ketbra xy := \ket x \circ \bra y$.

We give in the following the definition of Sum-Over-Paths of \cite{SOP-Clifford}, which differs from \cite{SOP} in the way the input qubits are treated, by making them more symmetric with the outputs. This makes some concepts, like the $\dagger$ or the compact structure, more natural.

\begin{defi}[$\cat{SOP}$]
We define $\cat{SOP}$ as the collection of objects $\mathbb N$ and morphisms between them that are tuples $f:n\to m := (s, \vec y, P, \vec O, \vec I)$, which we write:
\[ s\sum\limits_{\vec y\in V^k} e^{2i\pi P(\vec y)}\ketbra{\vec O(\vec y)}{\vec I(\vec y)}\]
where $s\in\mathbb{R}$, $\vec y\in V^k$ with $V$ a set of variables, $P\in \mathbb R[X_1,\ldots,X_{k}]/(X_i^2-X_i)$ is called the \emph{phase polynomial} of $f$\footnote{The quotient in the phase polynomial means that we consider each occurrence of the square of a variable to be equal to the variable itself $X_i^2-X_i=0$, since they will be evaluated over $\{0,1\}$. We can further constrain the polynomial by taking it modulo $1$, but only when considered as an element of a group, once all the products have been evaluated, as otherwise all phase polynomials would be evaluated to $0$ as $P = P\times1 = P\times 0 = 0$.}, $\vec O \in \left(\mathbb F_2[X_1,\ldots,X_{k}]\right)^m$ and $\vec I \in \left(\mathbb F_2[X_1,\ldots,X_{k}]\right)^n$ -- $\mathbb F_2$ being the binary field, whose only two elements are its additive and multiplicative identities, denoted $0$ and $1$.

\noindent Compositions are obtained as\footnote{To avoid further clutter, we may not specify the variables of polynomials, e.g.~$P_g$ actually stands for $P_g(\vec y_g)$, $\vec O_g$ for $\vec O_g(\vec y_g)$ etc...}:
\begin{itemize}
\item $f\circ g := \frac{s_fs_g}{2^m}\sum\limits_{\substack{\vec y_f,\vec y_g\\\vec y\in V^m}} e^{2i\pi \left(P_g+P_f+\frac{\widehat{\vec O_g}\cdot \vec y+\widehat{\vec I_f}\cdot \vec y}{2}\right)}\ketbra{\vec O_f}{\vec I_g}$ where $m=\abs{\vec I_f}=\abs{\vec O_g}$
\item $f\otimes g := s_fs_g\sum\limits_{\substack{\vec y_f,\vec y_g}} e^{2i\pi (P_g+P_f)}\ketbra{\vec O_f\vec O_g}{\vec I_f\vec I_g}$
\end{itemize}

\noindent We distinguish particular morphisms:
\begin{itemize}
\item Identity morphisms $id_n : \sum\limits_{\vec y\in V^n}\ketbra{\vec y}{\vec y}$
\item Symmetric braidings $\sigma_{n,m}= \sum\limits_{\vec y_1,\vec y_2}\ketbra{\vec y_2,\vec y_1}{\vec y_1,\vec y_2}$
\item Morphisms for compact structure $\eta_n= \sum\limits_{\vec y} \ketbra{\vec y,\vec y}{}$ and $\epsilon_n= \sum\limits_{\vec y} \ketbra{}{\vec y,\vec y}$
\end{itemize}

\noindent We also distinguish two functors that have $\cat{SOP}$ as a domain:
\begin{itemize}
\item The $\dagger$-functor is given by: $f^\dagger :=  s\sum\limits_{\vec y} e^{-2i\pi P}\ketbra{\vec I}{\vec O}$
\item The functor $\interp{\cdot}:\cat{SOP}\to\cat{Qubit}$ is defined as: $\interp{f}:= s\sum\limits_{\vec y\in \{0,1\}^k} e^{2i\pi P(\vec y)}\ketbra{\vec O(\vec y)}{\vec I(\vec y)}$
\end{itemize}
\end{defi}

The $\dagger$-functor is particularly important to characterise maps that are unitary -- the pure transformations that are allowed by quantum mechanics: $f$ is called unitary if $\interp{f^\dagger\circ f}=id$.

\begin{rem}
In the sequential composition $(\cdot\circ\cdot)$, the term $\widehat{\vec O_g}\cdot\vec y$ is to be understood as
\[\widehat{\vec O_g}\cdot\vec y = \widehat{O_g}_1\cdot y_1 + ... + \widehat{O_g}_m\cdot y_m\]
where $\vec O_g = ({O_g}_1,...,{O_g}_m)$ and $\vec y = (y_1,...,y_m)$; and where the map $\widehat{(\cdot)}:\mathbb F_2[X_1,\ldots,X_{k}]\to \mathbb R[X_1,\ldots,X_{k}]/(X_i^2-X_i)$, which translates a boolean polynomial into a phase polynomial, is inductively defined as:
\[ \widehat{Q_1Q_2}=\widehat{Q_1}\widehat{Q_2}\qquad\qquad \widehat{Q_1\oplus Q_2}=\widehat{Q_1}+\widehat{Q_2}-2\widehat{Q_1}\widehat{Q_2} \qquad\qquad \widehat{y_i}=y_i \qquad\qquad \widehat{\alpha}=\alpha\]
This translation will also be used in most of the upcoming rewrite rules of the system.
\end{rem}

\begin{exa}
\label{ex:H-Tof}
The Hadamard and Toffoli gates (which justify the name of the first fragment we will consider in the following), can be represented in this formalism as:
\[H:=\frac1{\sqrt2}\sum_{y_0,y_1}e^{2i\pi\frac{y_0y_1}2}\ketbra{y_1}{y_0}\qquad\operatorname{Tof}:=\sum_{y_0,y_1,y_2}\ketbra{y_0,y_1,y_2\oplus y_0y_1}{y_0,y_1,y_2}\]
It can be checked that both operators are unitary. To illustrate the compositions of $\cat{SOP}$ terms, we can focus on the following quantum circuit:
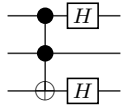
\[
\input{./figures/tof-H-compo-example.tikz}
\]
which is a graphical representation of $(H\otimes id\otimes H)\circ \operatorname{Tof}$ (spatial composition represents the tensor product, while sequential composition represents the usual composition). The term $H\otimes id\otimes H$ can be built as follows (renaming the variables so as to avoid collisions):
\begin{align*}
H\otimes id\otimes H& = \left({\textstyle\frac1{\sqrt2}}\!\sum_{y_0,y_1}e^{2i\pi\frac{y_0y_1}2}\ketbra{y_1}{y_0}\right)\!\!\otimes\!\!
\left(\sum_{y_2}e^{2i\pi\times 0}\ketbra{y_2}{y_2}\right)\!\!\otimes\!\!
\left({\textstyle\frac1{\sqrt2}}\!\sum_{y_3,y_4}e^{2i\pi\frac{y_3y_4}2}\ketbra{y_4}{y_3}\right)\\
&= \frac12 \sum_{y_0,...,y_4}e^{2i\pi(\frac{y_0y_1}2+\frac{y_3y_4}2)}\ketbra{y_1,y_2,y_4}{y_0,y_2,y_3}
\end{align*}
The $\cat{SOP}$ term $t=(H\otimes id\otimes H)\circ \operatorname{Tof}$ is then computed as:
\begin{align*}
t
&= \left(\frac12 \sum_{y_0,...,y_4}e^{2i\pi(\frac{y_0y_1}2+\frac{y_3y_4}2)}\ketbra{y_1,y_2,y_4}{y_0,y_2,y_3}\right)\circ \left(\sum_{y_5,y_6,y_7}\ketbra{y_5,y_6,y_7\oplus y_5y_6}{y_5,y_6,y_7}\right)\\
&= \frac1{2^4}\sum_{y_0,...,y_{10}}e^{2i\pi(\frac{y_0y_1}2+\frac{y_3y_4}2+y_8\frac{y_0+y_5}2+y_9\frac{y_2+y_6}2+y_{10}\frac{y_3+\widehat{y_7\oplus y_5y_6}}2)}\ketbra{y_1,y_2,y_4}{y_5,y_6,y_7}\\
&= \frac1{2^4}\sum_{y_0,...,y_{10}}e^{2i\pi(\frac{y_0y_1}2+\frac{y_3y_4}2+y_8\frac{y_0+y_5}2+y_9\frac{y_2+y_6}2+y_{10}\frac{y_3+y_7+ y_5y_6}2)}\ketbra{y_1,y_2,y_4}{y_5,y_6,y_7}
\end{align*}
where the last equality is obtained thanks to $y_{10}\frac{\widehat{y_7\oplus y_5y_6}}2=y_{10}(\frac{y_7+y_5y_6}2-y_5y_6y_7)=y_{10}\frac{y_7+y_5y_6}2$ when taken modulo $1$.
\end{exa}

As is customary, and as was already done in the previous example, we consider equality of the SOP morphisms up to $\alpha$-conversion, i.e.~renaming of the variables.
Notice that the definition of the composition $(\cdot\circ\cdot)$ gets somewhat involved. This is to cope with the way we deal with the inputs, which can be any boolean polynomial. The additional terms $\frac{\vec O_g\cdot \vec y+\vec I_f\cdot \vec y}{2}$ enforce that $O_{gi}=I_{fi}$ for all $0\leq i <m$. Indeed, when summing over the variable $y_i$, we get $(1+e^{i\pi(O_{gi}+I_{fi})})$ -- which is non-null only when $O_{gi}=I_{fi}$ -- as a factor of the whole morphism. This presentation has the advantage of keeping the size of the morphism polynomial with the size of the quantum circuit -- or ZH-diagram, see below -- it can be built from, no matter what gate set is used. A downside, however, is that the above does not directly constitute a category, as for instance $id\circ id\neq id$. However, it suffices to quotient the formalism with rewrite rules to turn it into a proper category \cite{SOP-Clifford}, hence justifying the use of the term ``functor'' for the last two maps.

\subsection{A Rewrite System}

We hence give in \autoref{fig:rewrite-TH} a set of rewrite rules denoted $\rewrite{\text{TH}}$ that induces an equational theory $\underset{\textnormal{TH}}{\sim}$ (the symmetric and transitive closure of $\rewrite{\text{TH}}$).

\begin{figure}[!ht]

\[\sum_{\vec y} e^{2i\pi P}\ketbra{\vec O}{\vec I}
\rewrite{y_0\notin \Var(P,\vec O,\vec I)} 2\sum_{\vec y\setminus\{y_0\}} e^{2i\pi P}\ketbra{\vec O}{\vec I}\tag{Elim}\]
\[t = \sum e^{2i\pi\left(\frac{y_0}{2}(y_i\widehat Q + \widehat{Q'} +1) + R\right)}\ket{\vec O}\!\!\bra{\vec I}
\rewrite{y_0\notin\Var(Q,Q',R,\vec O,\vec I)\\y_i\notin\Var(Q,Q')\\QQ'=Q'} t[y_i\leftarrow 1\oplus Q']\tag{HHgen}\]
\[t=\sum_{\vec y} e^{2i\pi\left(P\right)}|\cdots,\overset{O_i}{\overbrace{y_0{\oplus} O_i'}},\cdots\rangle \!\!\bra{\vec I}
\rewrite{O_i'\neq 0\\y_0\notin \Var(O_1,\ldots,O_{i-1},O_i')}
t[y_0{\leftarrow} O_i]\tag{ket}\]
\vspace*{-0.5em}
\[t=\sum_{\vec y} e^{2i\pi\left(P\right)}\ket{\vec O}\!\!\langle \cdots,\overset{I_i}{\overbrace{y_0{\oplus} I_i'}},\cdots|
\rewrite{I_i'\neq 0\\y_0\notin \Var(\vec O,I_1,\ldots,I_{i-1},I_i')}
t[y_0{\leftarrow} I_i]\tag{bra}\]
\[s\sum_{\vec y} e^{2i\pi\left(\frac{y_0}{2} + R\right)}\ketbra{\vec O}{\vec I}
\rewrite{R\neq 0 \text{ or } \vec O,\vec I\neq \vec 0\\y_0\notin\Var(R,\vec O,\vec I)} \sum_{y_0} e^{2i\pi\left(\frac{y_0}{2}\right)}\ketbra{0,\cdots,0}{0,\cdots,0}\tag{Z}\]
\caption[]{Rewrite system $\rewrite{\text{TH}}$}
\label{fig:rewrite-TH}
\end{figure}

We need in the conditions of all the rules the function $\Var$, that, given a set or list of polynomials, gives the set of all variables used in them. 
We call \emph{internal variable} a variable that is present in the morphism $t$ but not in its inputs/outputs, i.e.~a variable $y_0$ such that $y_0\in\Var(t)\setminus\Var(\vec O,\vec I)$. It is worth noting that searching for an occurrence of, and applying any of these rules \emph{once}, can be done in polynomial time.

The rules (ket) and (bra) correspond to changes of variables that are necessary to get a unique normal form in the Clifford case \cite{SOP-Clifford}, and the rule (Elim) simply gets rid of a variable that is used nowhere in the term and simply contributes to a global phase (since that variable is supposed to range over two values, it contributes to a multiplicative scalar $2$).

The cornerstone rule of \cite{SOP} and \cite{SOP-Clifford} denoted (HH) was given as:
\[t = \sum e^{2i\pi\left(\frac{y_0}{2}(y_i + \widehat{Q}) + R\right)}\ketbra{\vec O}{\vec I}
\rewrite{y_0\notin\Var(Q,R,\vec O,\vec I)\\y_i\notin\Var(Q)} t[y_i\leftarrow Q]\tag{HH}\]
It has been generalised into (HHgen) here ((HH) is the special case where $Q=1$). It is important to notice that the rule (HHgen) requires that $QQ'=Q$. Rule (HH) is so often used that we may, in the following, distinguish it from (HHgen).

\begin{exa}
Recall the term:
\[t = \frac1{2^4}\sum_{y_0,...,y_{10}}e^{2i\pi(\frac{y_0y_1}2+\frac{y_3y_4}2+y_8\frac{y_0+y_5}2+y_9\frac{y_2+y_6}2+y_{10}\frac{y_3+y_7+ y_5y_6}2)}\ketbra{y_1,y_2,y_4}{y_5,y_6,y_7}\]
we got from \autoref{ex:H-Tof}. Its internal variables are $y_0,y_3,y_8,y_9, y_{10}$. It is often the case that the variables created by the sequential composition ($y_8,y_9$ and $y_{10}$ here), can be used right away to rewrite the term using (HH). For instance, the term $\frac{y_8}2(y_0+y_5)$ where $y_8$ appears nowhere else in the morphism, allows for an application of the rule, that will replace $y_0$ by $y_5$ (or vice-versa). We get:
\begin{align*}
t\rewrite{\text{HH}}\frac1{2^4}\sum_{y_1,...,y_{10}}e^{2i\pi(\frac{y_5y_1}2+\frac{y_3y_4}2+y_9\frac{y_2+y_6}2+y_{10}\frac{y_3+y_7+ y_5y_6}2)}\ketbra{y_1,y_2,y_4}{y_5,y_6,y_7}
\end{align*}
The rule (HH) can always be followed by (Elim): here $y_8$ remains in the variables, but isn't used anywhere anymore. $t$ can then be further reduced to:
\[\frac1{2^3}\sum_{\substack{y_i\\i\in\{1,2,3,4,5,6,7,9,10\}}}e^{2i\pi(\frac{y_5y_1}2+\frac{y_3y_4}2+y_9\frac{y_2+y_6}2+y_{10}\frac{y_3+y_7+ y_5y_6}2)}\ketbra{y_1,y_2,y_4}{y_5,y_6,y_7}\]
Proceeding similarly for $y_9$ and $y_{10}$, we get:
\begin{align*}
t&\rewrite{}^*\frac1{2^2}\sum_{\substack{y_i\\i\in\{1,3,4,5,6,7,10\}}}e^{2i\pi(\frac{y_5y_1}2+\frac{y_3y_4}2+y_{10}\frac{y_3+y_7+ y_5y_6}2)}\ketbra{y_1,y_6,y_4}{y_5,y_6,y_7}\\
&\rewrite{}^*\frac1{2}\sum_{\substack{y_i\\i\in\{1,4,5,6,7\}}}e^{2i\pi(\frac{y_5y_1}2+\frac{y_4y_7}2+ \frac{y_4y_5y_6}2)}\ketbra{y_1,y_6,y_4}{y_5,y_6,y_7}
\end{align*}
The term cannot be reduced further with the rewrite system $\rewrite{\text{TH}}$.
\end{exa}
In the following, we assume that (Elim) is always applied after (HH), without necessarily mentioning it.

The rules (HH), (HHgen) and (Z) all stem from a particular observation: In the morphism $t = \sum e^{2i\pi\left(\frac{y_0}{2}\widehat Q + R\right)}\ketbra{\vec O}{\vec I}$ where $y_0$ is internal and not in $R$, if $Q$ is evaluated to $1$, then the whole morphism is interpreted as null. This is exactly what (Z) captures -- and the conditions on $R$, $\vec O$ and $I$ are simply here to avoid applying the rule indefinitely.

The rule (HH) deals with a case where the polynomial $Q$ can be forced to $0$, whilst the rule (HHgen) more generally deals with a case where one of the variables in the polynomial $Q$ is forced to get a precise value due to the form of the polynomial.

The rule (HH) is of particular interest. It was introduced in \cite{SOP} and gives enough power to the formalism to become a $\dagger$-compact PROP \cite{SOP-Clifford}. We can extend this result here thanks to:
\begin{prop}
\label{prop:TH-local}
\[\forall t_1,t_2 \in \cat{SOP},~~t_1\underset{\textnormal{TH}}{\sim}t_2 \implies \begin{cases}
A\circ t_1 \circ B \underset{\textnormal{TH}}{\sim} A\circ t_2 \circ B & \text{for all $A$, $B$ composable}\\
A\otimes t_1 \otimes B \underset{\textnormal{TH}}{\sim} A\otimes t_2 \otimes B & \text{for all $A$, $B$}
\end{cases}\]
\end{prop}

\begin{proof}
\phantomsection\label{prf:TH-local}
The result is obvious for the tensor product $(.\otimes.)$. For the composition, we show that if $t_1\rewrite{\text{TH}}t_2$ in one step, then $A\circ t_1\circ B \underset{\textnormal{TH}}{\sim} A\circ t_2\circ B$. In other words, we have to show it for every rule in $\rewrite{\text{TH}}$:

\noindent$\bullet$ (Elim): Obvious.

\noindent$\bullet$ (HHgen):
\begin{align*}
A\circ t_1\circ B = \sum e^{2i\pi \left(P_A+P_B+\frac{y_0}2(y_i\widehat Q + \widehat{Q'} +1) + R + \frac{\vec O\cdot \vec x + \vec I_A\cdot\vec x + \vec I\cdot \vec x' + \vec O_B\cdot \vec x'}2\right)} \ketbra{\vec O_A}{\vec I_B}\\
\rewrite{\text{HHgen}}
\sum e^{2i\pi \left(P_A+P_B+\frac{y_0}2(\widehat Q +1) + R[y_i\leftarrow \widehat{1{\oplus}Q'}] + \frac{\vec O[y_i\leftarrow \widehat{1{\oplus}Q'}]\cdot \vec x + \vec I_A\cdot\vec x + \vec I[y_i\leftarrow \widehat{1{\oplus}Q'}]\cdot \vec x' + \vec O_B\cdot \vec x'}2\right)} \raisebox{-2ex}{\hspace*{-2em}$\ketbra{\vec O_A}{\vec I_B}$}\\
= A\circ t_1[y_i\leftarrow {1{\oplus}Q'}]\circ B
= A\circ t_2\circ B
\end{align*}

\noindent$\bullet$ (ket):
\begin{align*}
A\circ t_1\circ B =\hspace*{32em}\\
\sum e^{2i\pi \left(P_A+P_B+P+ \frac{(\widehat O_1 +\widehat I_{A1})x_1 +\ldots + (y_0+\widehat O_i'+\widehat I_{Ai})x_i+\ldots +(\widehat O_m +\widehat I_{Am})x_m + \vec I\cdot \vec x' + \vec O_B\cdot \vec x'}2\right)} \ketbra{\vec O_A}{\vec I_B}\\
\rewrite{\text{HH}}
2\sum e^{2i\pi \left(P_A+P_B+P[y_0\leftarrow \widehat{O_i'{\oplus}I_{Ai}}] + \frac{\vec O[y_0\leftarrow \widehat{O_i'{\oplus}I_{Ai}}]\cdot \vec x + \vec I_A\cdot\vec x + \vec I[y_0\leftarrow \widehat{O_i'{\oplus}I_{Ai}}]\cdot \vec x' + \vec O_B\cdot \vec x'}2\right)} \ketbra{\vec O_A}{\vec I_B}\\
\underset{\text{HH}}\longleftarrow
\sum e^{2i\pi \left(P_A+P_B+P[y_0\leftarrow \widehat{y_0{\oplus}O_i'}] + \frac{\vec O[y_0\leftarrow \widehat{y_0{\oplus}O_i'}]\cdot \vec x + \vec I_A\cdot\vec x + \vec I[y_0\leftarrow \widehat{y_0{\oplus}O_i'}]\cdot \vec x' + \vec O_B\cdot \vec x'}2\right)} \ketbra{\vec O_A}{\vec I_B}\\
=A\circ t_1[y_0\leftarrow {y_0{\oplus}O_i'}]\circ B = A\circ t_2\circ B
\end{align*}

\noindent$\bullet$ (bra): Similar to (ket).

\noindent$\bullet$ (Z):
\begin{align*}
A\circ t_1\circ B = \sum &e^{2i\pi \left(P_A+P_B+\frac{y_0}2 + R + \frac{\vec O\cdot \vec x + \vec I_A\cdot\vec x + \vec I\cdot \vec x' + \vec O_B\cdot \vec x'}2\right)} \ketbra{\vec O_A}{\vec I_B}
\rewrite{\text{Z}}
\sum e^{2i\pi \left(\frac{y_0}2\right)} \ketbra{\vec 0}{\vec 0}\\
&\underset{\text{Z}}\longleftarrow
\sum e^{2i\pi \left(P_A+P_B+\frac{y_0}2 + \frac{\vec 0\cdot \vec x + \vec I_A\cdot\vec x + \vec 0\cdot \vec x' + \vec O_B\cdot \vec x'}2\right)} \ketbra{\vec O_A}{\vec I_B}
= A\circ t_2\circ B \qedhere
\end{align*}
\end{proof}

Thanks to this Proposition, and since $\cat{SOP}/\underset{\textnormal{HH}}{\sim}$ is a $\dagger$-compact PROP by \cite{SOP-Clifford}, we get:
\begin{cor}
$\cat{SOP}/\underset{\textnormal{TH}}{\sim}$ is a $\dagger$-compact PROP.
\end{cor}

The set of rules was obviously chosen so as to preserve the semantics:
\begin{prop}[Soundness]
\label{prop:TH-soundness}
For any two $\cat{SOP}$ morphisms $t_1$ and $t_2$, if $t_1\overset{\ast}{\rewrite{\textnormal{TH}}}t_2$, then $\interp{t_1}=\interp{t_2}$.
\end{prop}

\begin{proof}
\phantomsection\label{prf:TH-soundness}
We mean to show that for any single step rewrite, the interpretation is preserved. As most of the rules were present in \cite{SOP} or \cite{SOP-Clifford}, this was proven for them. It remains to show the result for (HHgen). The soundness of a stronger version of (HHgen) is proven in upcoming \autoref{lem:hhgen'-sound}. The result is then a direct consequence of this.
\end{proof}

The addition of the rule (HHgen) allows for further reductions:
\begin{exa}
\label{ex:HHgen}
The following morphism:
\[\sum_{\vec y}e^{2i\pi(\frac{y_0y_1y_2}2+\frac{y_2}2+\frac{y_1y_2y_3}2+\frac{y_0y_1y_2y_3}2)}\ketbra{y_3}{y_0}\]
is irreducible using the rules of \cite{SOP} and \cite{SOP-Clifford}. However, here it can be reduced to:
\begin{align*}
\sum_{\vec y}e^{2i\pi(\frac{y_0y_1y_2}2+\frac{y_2}2+\frac{y_1y_2y_3}2+\frac{y_0y_1y_2y_3}2)}\ketbra{y_3}{y_0}
\overset{\text{(HHgen)}}{\rewrite{y_1\leftarrow 1}}
\sum_{y_0,y_2,y_3}e^{2i\pi(\frac{y_0y_2}2+\frac{y_2}2+\frac{y_2y_3}2+\frac{y_0y_2y_3}2)}\ketbra{y_3}{y_0}
\end{align*}
where the rewrite can be made clearer by writing the phase polynomial as $\frac{y_2}2(y_1(y_0+y_3+y_2y_3)+0+1)$ and checking that (obviously) $(y_0+y_3+y_2y_3)\times 0 = 0$.
\end{exa}

\subsection{Generalising Rules}

In \cite{Vilmart2023completeness} another rule was introduced to allow for the proof of completeness:
\[t = \sum e^{2i\pi\left(\frac{y_0}{2}\widehat Q + \frac{y_0'}{2}\widehat{Q'} + R\right)}\ketbra{\vec O}{\vec I}
\rewrite{y_0,y_0'\notin \Var(Q,Q',R,\vec O,\vec I)} 2t[y_0'\leftarrow y_0\oplus y_0Q]\tag{HHnl}\]
This rule was only used when proving that one of the axioms of $\operatorname{ZH}_{\operatorname{TH}}$ from \cite{phase-free-ZH} was provable with $\underset{\operatorname{TH}}\sim$. However, it was proven in \cite{backens2021ZHcompleteness} that this axiom could be replaced by two simpler ones and otherwise provable, hence making (HHnl) unnecessary for completeness, as will be shown in the following of the paper.

Completeness however assumes rules can be used both ways, while we ideally only want to reduce a term when analysing it. With this in mind, it is relevant to look for rules that allow for further reduction, and (HHnl) is one of them.

\begin{exa}
The following term, while irreducible with $\underset{\text{TH}}\longrightarrow$, can be reduced using (HHnl):
\begin{align*}
\sum_{y_0,y_1,y_2,y_3,y_4}e^{2i\pi(\frac{y_0y_1y_2}2+\frac{y_2}2+\frac{y_2y_3y_4}2)}&\ketbra{y_4}{y_0}\\
\overset{\text{(HHnl)}}{\rewrite{y_3\leftarrow y_1\oplus y_0y_1y_2}}
&~2\sum_{y_0,y_1,y_2,y_4}e^{2i\pi(\frac{y_0y_1y_2}2+\frac{y_2}2+\frac{y_1y_2y_4}2+\frac{y_0y_1y_2y_4}2)}\ketbra{y_4}{y_0}
\end{align*}
\end{exa}

As (HHnl) is not necessary for completeness, it should be possible to derive it from the rules of $\underset{\text{TH}}\longrightarrow$:

\begin{lem}
(HHnl) can be derived from $\underset{\text{TH}}\longrightarrow$.
\end{lem}

\begin{proof}
Consider the term:
\[\sum e^{2i\pi \left(\frac{y_0}2 + \frac{y_0y_1(\widehat Q+1)}2 + \frac{y_1y_0'}2 + \frac{y_0'(\widehat{Q'}+1)}2 + R \right)}\ketbra{\vec O}{\vec I}\]
We can rewrite it in two different ways, which give both sides of the (HHnl) rule:
\begin{gather*}
\sum e^{2i\pi \left(\frac{y_0}2 + \frac{y_0y_1(\widehat Q+1)}2 + \frac{y_1y_0'}2 + \frac{y_0'(\widehat{Q'}+1)}2 + R \right)}\ketbra{\vec O}{\vec I}\\
\begin{array}[b]{cc}
\phantom{\text{HHgen}(y_1\leftarrow1)}\Big\downarrow\text{HHgen}(y_1\leftarrow1)&\phantom{\text{HH}(y_1\leftarrow Q'\oplus1)}\Big\downarrow\text{HH}(y_1\leftarrow Q'\oplus1)\\[0.7em]
\sum e^{2i\pi \left(\frac{y_0}2\widehat Q+ \frac{y_0'}2\widehat{Q'} + R\right)}\ketbra{\vec O}{\vec I} &\qquad
2\sum e^{2i\pi \left(\frac{y_0}2(\widehat Q + \widehat{Q'} + \widehat{QQ'}) + R\right)}\ketbra{\vec O}{\vec I}
\end{array}\qedhere
\end{gather*}
\end{proof}

When rewriting SOP-morphisms for simplification or verification, it can be beneficial to not only reduce the number of variables -- which is what all rules but (ket/bra) do --, but also to keep the size of the phase polynomial as short as possible. In that respect, the rule (HHgen) itself can be generalised to:
\[t = \sum e^{2i\pi\left(\frac{y_0}{2}(y_i\widehat Q + \widehat{QQ'} +1) + R\right)}\ketbra{\vec O}{\vec I}
\rewrite{y_0\notin\Var(Q,Q',R,\vec O,\vec I)\\y_i\notin\Var(Q,Q')} t[y_i\leftarrow 1\oplus Q']\tag{HHgen'}\]
where the polynomial $Q'$ can here be smaller (in the number of terms) than the one in (HHgen). However, finding a ``minimal'' $Q'$ for this rule is a hard problem, as it requires the use of boolean Groebner bases \cite{Boolean-Grobner-Bases}, while instances (HHgen) can easily be found. (HHgen) can be seen as a particular case of (HHgen'), where $Q'\leftarrow QQ'$, as $Q\times QQ' = QQ'$.

The previous observation can be made into another rule, which again uses the fact that when a term has a phase polynomial of the form $\frac{y_0}2 \widehat Q + R$ with $y_0$ internal, then we can force $Q$ to be $0$. This means we can replace $R$ by a remainder in the euclidean division of $R$ by $Q$:

\[\sum e^{2i\pi\left(\frac{y_0}{2}\widehat{Q} + \widehat{SQ} + \widehat R\right)}\ketbra{\vec O}{\vec I}
\rewrite{y_0\notin\Var(Q,S,R,\vec O,\vec I)}
\sum e^{2i\pi\left(\frac{y_0}{2}\widehat{Q} +  \widehat R\right)}\ketbra{\vec O}{\vec I}\tag{Rem}\]

Notice that, in contrast with the (HH), (HHgen) and (HHgen'), this rule does not reduce the number of variables, but instead equates polynomials that are equivalent ``modulo $Q$''. In practice, this rule is also hard to use as it again uses Groebner bases to properly implement. The proximity between the last two aforementioned rules is not surprising once we realise that (Rem) can be deduced from (HHgen) and (HHgen'):
\begin{gather*}
\sum e^{2i\pi \left(\frac{y_0}2\left(y_i(1+\widehat Q)+(1+\widehat Q)\widehat S+1\right)+ \frac{\widehat{S}}2 +\frac{y_i}2 +\frac12 + R \right)}\ketbra{\vec O}{\vec I}\\
\begin{array}{cc}
\text{HHgen}(y_i\leftarrow (1\oplus Q)S\oplus 1)\Big\downarrow\qquad\qquad&\qquad\qquad\Big\downarrow\text{HHgen'}(y_i\leftarrow S\oplus 1)\\[0.7em]
\sum e^{2i\pi \left(\frac{y_0}2\widehat Q+ \widehat{SQ} + R\right)}\ketbra{\vec O}{\vec I} &\qquad
\sum e^{2i\pi \left(\frac{y_0}2\widehat Q + R\right)}\ketbra{\vec O}{\vec I}
\end{array}
\end{gather*}

We can show that (HHgen') is sound, which implies that (HHgen) -- as a particular case -- and both (HHnl) and (Rem) -- as a compositions of particular cases of (HHgen') -- are also sound:

\begin{lem}
\label{lem:hhgen'-sound}
(HHgen') is sound.
\end{lem}

\begin{proof}
If $t = \sum_{\vec y\in V^k} e^{2i\pi\left(\frac{y_0}{2}(y_i\widehat Q + \widehat{QQ'} +1) + R\right)}\ketbra{\vec O}{\vec I}$, then:
\begin{align*}
\interp{t}
&= \sum_{\vec y\in \{0,1\}^k} e^{2i\pi\left(\frac{y_0}{2}(y_i\widehat Q + \widehat{QQ'} +1) + R\right)}\ketbra{\vec O}{\vec I}
= \sum_{\vec y\in \{0,1\}^{k-1}} (1+e^{i\pi (y_i\widehat Q + \widehat{QQ'} +1)})e^{2i\pi R}\ketbra{\vec O}{\vec I}\\
&= \sum_{\vec y\in \{0,1\}^{k-2}} (1+e^{i\pi (\widehat{Q'}\widehat Q + \widehat{QQ'} +1)})e^{2i\pi R[y_i\leftarrow \widehat Q']}(\ketbra{\vec O}{\vec I})[y_i\leftarrow Q']\\[-0.8em]
&\tag*{$\displaystyle + \sum_{\vec y\in \{0,1\}^{k-2}} (1+e^{i\pi (\widehat{1\oplus Q'}\widehat Q + \widehat{QQ'} +1)})e^{2i\pi R[y_i\leftarrow \widehat{1\oplus Q'}]}(\ketbra{\vec O}{\vec I})[y_i\leftarrow 1\oplus Q']$}\\
&= 0 +\sum_{\vec y\in \{0,1\}^{k-2}} (1+e^{i\pi (\widehat{1\oplus Q'}\widehat Q + \widehat{QQ'} +1)})e^{2i\pi R[y_i\leftarrow \widehat{1\oplus Q'}]}(\ketbra{\vec O}{\vec I})[y_i\leftarrow 1\oplus Q']\\
&=\interp{t[y_i\leftarrow 1\oplus Q']}\qedhere
\end{align*}
\end{proof}

While we have identified stronger rewrite rules than the ones in $\underset{\operatorname{TH}}\sim$, as the rest of the paper is mainly focussed on completeness, we stick to the rules of $\underset{\operatorname{TH}}\sim$, since, as we shall see later, they are enough for that particular problem.

\section{The $\cat{ZH}$-Calculus}
\label{sec:ZH}

As a foundation towards completeness of the Toffoli-Hadamard fragment of $\cat{SOP}$, we will use a similar result on another formalism: the graphical calculus ZH.

The graphical calculi ZX, ZW and ZH \cite{ZH,interacting,ghz-w} are calculi for quantum computing, with a tight link with the Sum-Over-Paths formalism \cite{MSc.Lemonnier,LvdWK,SOP-Clifford}, and whose completeness was proven in particular for the Toffoli-Hadamard fragment \cite{backens2021ZHcompleteness,zw,phase-free-ZH,zx-toffoli}. 

This fragment of quantum mechanics is approximately universal \cite{toffoli-simple,toffoli}, and it is arguably the simplest one with this property. This is the fragment we will be interested in, in most of the following of the paper; and the associated completeness result will be paramount in the development of the following.

We choose to present here the ZH-calculus, because of its proximity with both $\cat{SOP}$ and the Toffoli-Hadamard fragment. Notice however that there exist translations between all the aforementioned graphical calculi, so by composition, we can connect $\cat{SOP}$ to all of them.

\subsection{The Diagrams and their Interpretation}

$\cat{ZH}$ is a PROP whose morphisms -- read here from top to bottom -- are composed (sequentially $(\cdot\circ\cdot)$ or in parallel $(\cdot\otimes\cdot)$) from Z-spiders and H-spiders:
\begin{itemize}
\item $Z_m^n:n\to m::$
\input{./figures/Z-spider.tikz}
, called Z-spider
\item $H_m^n(r):n\to m::$
\input{./figures/H-spider.tikz}
, called H-spider, with a parameter $r\in\mathbb C$
\end{itemize}
When $r$ is not specified, the parameter in the H-spider is taken to be $-1$.

$\cat{ZH}$ is made a $\dagger$-compact PROP, which means it also has a symmetric structure $\sigma_{n,m}::
\input{./figures/sigma.tikz}
$, a compact structure $\left(\eta_n::
\input{./figures/eta.tikz}
,\epsilon_n::
\input{./figures/epsilon.tikz}
\right)$, and a $\dagger$-functor $(\cdot)^\dagger:\cat{ZH}^{\operatorname{op}}\to\cat{ZH}$. It is defined by: $(Z_m^n)^\dagger := Z_n^m$ and $(H_m^n(r))^\dagger := H_n^m(\overline{r})$ 
where $\overline r$ is the complex conjugate of $r$. 
For convenience, we define two additional spiders:\\

\input{./figures/X-spider.tikz}
 and 
\input{./figures/X-spider-neg.tikz}

The language comes with a way of interpreting the morphisms as morphisms of $\cat{Qubit}$. The standard interpretation $\interp{\cdot}:\cat{ZH}\to\cat{Qubit}$ is a $\dagger$-compact-PROP-functor, defined as:
\[\interp{
\input{./figures/Z-spider.tikz}
} = \ketbra{0^m}{0^n} + \ketbra{1^m}{1^n}
\qquad\qquad\interp{~
\begin{tikzpicture}
	\begin{pgfonlayer}{nodelayer}
		\node [style=none] (0) at (0, 0.25) {};
		\node [style=none] (1) at (0, -0.25) {};
	\end{pgfonlayer}
	\begin{pgfonlayer}{edgelayer}
		\draw (0.center) to (1.center);
	\end{pgfonlayer}
\end{tikzpicture}

~}=\ketbra00+\ketbra11\]
\[\interp{
\input{./figures/H-spider.tikz}
} = \sum_{j_k,i_k\in\{0,1\}}r^{j_1\ldots j_mi_1\ldots i_n}\ketbra{j_1,\ldots,j_m}{i_1,\ldots,i_n}\]
\[\interp{
\input{./figures/eta.tikz}
} = \sum_{i_k\in\{0,1\}}\ket{i_1,\ldots,i_n,i_1,\ldots,i_n}=\interp{
\input{./figures/epsilon.tikz}
}^\dagger\]
\[\interp{
\input{./figures/sigma.tikz}
} = \sum_{i_k,j_k\in\{0,1\}}\ketbra{j_1,\ldots,j_m,i_1,\ldots,i_n}{i_1,\ldots,i_n,j_1,\ldots,j_m}\]

Notice that we used the same symbol for two different functors: the two interpretations $\interp{\cdot}:\cat{SOP}\to\cat{Qubit}$ and $\interp{\cdot}:\cat{ZH}\to\cat{Qubit}$. It should be clear from the context which one is to be used.\\
The language is universal: $\forall f\in \cat{Qubit},~\exists D_f\in\cat{ZH},~~\interp{D_f} = f$. In other words, the interpretation $\interp{\cdot}$ is surjective.

\subsection{Equational Theory}

The language comes with an equational theory, which in particular gives the axioms for a $\dagger$-compact PROP. We will not present it here.

We can easily define a restriction of $\cat{ZH}$ that exactly captures the Toffoli-Hadamard fragment of quantum mechanics \cite{backens2021ZHcompleteness,phase-free-ZH}, as the language generated by:
\[\left\lbrace
\input{./figures/H-spider-minus-1.tikz}
,
\input{./figures/Z-spider.tikz}
,
\begin{tikzpicture}
	\begin{pgfonlayer}{nodelayer}
		\node [style=box] (0) at (0, 0) {$\frac1{\sqrt2}$};
	\end{pgfonlayer}
\end{tikzpicture}

\right\rbrace.\]
Notice that the two black spiders can still be defined if we also define $
\begin{tikzpicture}
	\begin{pgfonlayer}{nodelayer}
		\node [style=box] (0) at (0.125, 0) {$\frac1{\sqrt2^p}$};
	\end{pgfonlayer}
\end{tikzpicture}

:=
\begin{tikzpicture}
	\begin{pgfonlayer}{nodelayer}
		\node [style=box] (0) at (0, 0) {$\frac1{\sqrt2}$};
	\end{pgfonlayer}
\end{tikzpicture}

^{\otimes p}$. We denote this restriction by $\cat{ZH}_{\operatorname{TH}}$.

This restriction is provided with an equational theory, given in \autoref{fig:ZH-TH-rules}, that makes it complete\footnote{The axiomatisation provided here is based on that of \cite{backens2021ZHcompleteness} together with \cite[Thm.~8.6]{backens2021ZHcompleteness} which simplifies one of the rules. It simplifies the axiomatisation found in \cite{phase-free-ZH}, however the fragment considered is slightly different, as \cite{backens2021ZHcompleteness} does \emph{not} contain the scalar $\frac1{\sqrt2}$, but only $\frac12$. It is however very easy to extend the completeness to the language obtained by adjoining a generator for the $\frac1{\sqrt2}$ scalar, using the additional rule that states that $2*\frac{1}{\sqrt2}*\frac{1}{\sqrt2}=1$. This is what is done in the axiomatisation provided here. A proof that completeness does extend can be adapted from e.g.~\cite[Prop.~1]{JPV}.}.

\begin{figure}[!htb]

\input{./figures/ZH-rule-ZS1.tikz}
\hfill
\input{./figures/ZH-rule-ZS2.tikz}
\hfill
\input{./figures/ZH-rule-HS1.tikz}
\hfill
\input{./figures/ZH-rule-HS2.tikz}
\\[1.5em]
\phantom{.}\hfill
\input{./figures/ZH-rule-BA1.tikz}
\hfill
\input{./figures/ZH-rule-BA2.tikz}
\hfill
\input{./figures/ZH-rule-M.tikz}
\hfill\phantom{.}\\[1.5em]

\input{./figures/ZH-rule-CPH.tikz}
\hfill
\input{./figures/ZH-rule-IP.tikz}

\hfill
\input{./figures/ZH-rule-IV.tikz}
\hfill
\input{./figures/ZH-rule-Z.tikz}
\\[1.5em]
\phantom{.}\hfill
\input{./figures/X-spider.tikz}
\hfill
\input{./figures/X-spider-neg.tikz}
\hfill$
\input{./figures/H-scalar-1_sqrt2_p.tikz}
:=
\input{./figures/H-scalar-1_sqrt2.tikz}
^{\otimes p}$\hfill\phantom{.}
\caption[]{Set of rules $\operatorname{ZH}_{\operatorname{TH}}$ \cite{backens2021ZHcompleteness}.}
\label{fig:ZH-TH-rules}
\end{figure}

\begin{thmC}[\cite{backens2021ZHcompleteness} Completeness of $\cat{ZH}_{\operatorname{TH}}/\operatorname{ZH}_{\operatorname{TH}}$]
\label{thm:ZH-completeness}
\[\forall D_1,D_2\in \cat{ZH}_{\operatorname{TH}},~~\interp{D_1}=\interp{D_2}\iff \operatorname{ZH}_{\operatorname{TH}}\vdash D_1=D_2\]
\end{thmC}

\section{Translations between $\cat{SOP}$ and $\cat{ZH}$}
\label{sec:SOP<->ZH}

\subsection{From $\cat{SOP}$ to $\cat{ZH}$}

It is possible to translate $\cat{SOP}$ morphisms to ZH-diagrams using interpretation $[\cdot]^{\operatorname{ZH}}:\cat{SOP}\to\cat{ZH}$. 
A description of $[\cdot]^{\operatorname{ZH}}:\cat{SOP}\to\cat{ZH}$ was defined in \cite{MSc.Lemonnier,LvdWK} and in \cite{SOP-Clifford}. We choose the latter definition as it fits our definition of $\cat{SOP}$.
\[\left[s \sum_{\vec y} e^{2i\pi P}\ketbra{O_1,\ldots,O_m}{I_1,\ldots,I_n}\right]^{\operatorname{ZH}}:=
\input{./figures/ZH-NF.tikz}
\]
where the row of Z-spiders represents the variables $y_1,\ldots,y_k$.
Informally:
\begin{itemize}
\item each monomial $\alpha y_{i_1}...y_{i_s}$ in $P$ gives a single H-spider with parameter $e^{i\frac\alpha{2\pi}}$ and connected to the Z-spiders that represent $y_{i_1}$,...,$y_{i_s}$
\item each monomial $y_{i_1}...y_{i_s}$ in $O_i$ is represented by 
\input{./figures/and-n.tikz}
 where the inputs are connected to the Z-spiders that represent $y_{i_1}$,...,$y_{i_s}$. Notice that the only (non-zero) constant monomial is 
\input{./figures/constant-monomial.tikz}

\item these monomials are then added to form $O_i$ thanks to 
\input{./figures/xor-n.tikz}

\item the nodes $I_i$ are defined similarly, but upside-down
\end{itemize}
For more details, see \cite{SOP-Clifford}.

\begin{exa}
The $\cat{SOP}$ morphism:
\[\frac{1}{2\sqrt{2}}\sum\limits_{\vec y}e^{2i\pi\left(\frac14y_0+\frac{1}{2}y_4y_0+\frac{1}{8}y_5y_0y_1+\frac{3}{4}y_1y_2y_3+\frac{1}{2}y_0y_3\right)}\ketbra{0,1{\oplus} y_0{\oplus} y_4y_2,y_5}{y_4,y_5{\oplus} y_2{\oplus} y_3}\]
is mapped to 
\input{./figures/SOP-to-ZH-example.tikz}

\end{exa}
The boolean polynomials as defined above are given in their (unique) expanded form. These can easily be shown to be copied through the white node:
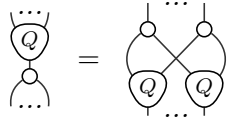
\begin{lem}
\label{lem:poly-copy}
\def\fig{polynomial-copy}
\begin{align*}
\input{./figures/\fig/\fig_00.tikz}
\eq{}\input{./figures/\fig/\fig_05.tikz}
\end{align*}
\end{lem}

\begin{proof}
\phantomsection\label{prf:poly-copy}
We prove the result where $Q$'s constant is $0$. The proof where $Q$'s constant is $1$ is very similar, and simply uses the fact that $
\begin{tikzpicture}
	\begin{pgfonlayer}{nodelayer}
		\node [style=none] (0) at (0, -0.375) {};
		\node [style=dot] (1) at (0, 0) {};
		\node [style=none, font={\scriptsize}] (2) at (-0.25, -0.025) {$\neg$};
		\node [style=none] (3) at (0, 0.375) {};
	\end{pgfonlayer}
	\begin{pgfonlayer}{edgelayer}
		\draw (1) to (0.center);
		\draw (3.center) to (1);
	\end{pgfonlayer}
\end{tikzpicture}

$ is copied through the white node \cite{backens2021ZHcompleteness}:
\def\fig{polynomial-copy}
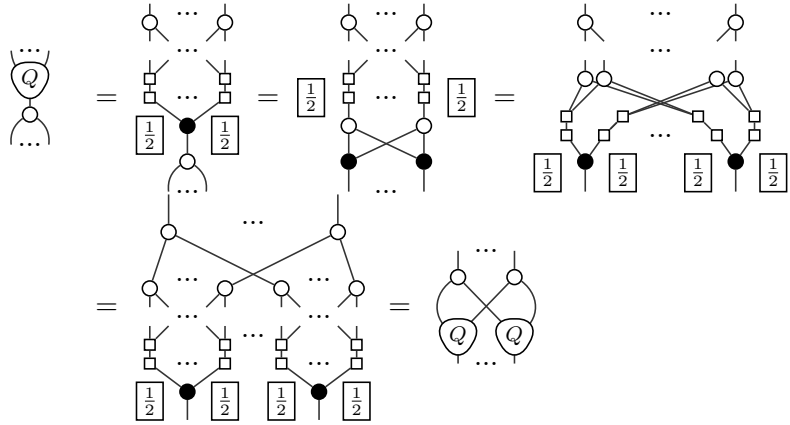
\[
  \ensuremath{
  \begin{array}[b]{rl}
\input{./figures/\fig/\fig_00.tikz}
&\eq{}\input{./figures/\fig/\fig_01.tikz}
\eq{}\input{./figures/\fig/\fig_02.tikz}
\eq{}\input{./figures/\fig/\fig_03.tikz}\\
&\eq{}\input{./figures/\fig/\fig_04.tikz}
\eq{}\input{./figures/\fig/\fig_05.tikz}
\\[-\baselineskip]
\end{array}}
\qedhere
\]
\end{proof}

\noindent The above translation preserves the semantics:
\begin{propC}[\cite{SOP-Clifford}]
\label{prop:zh-preserves-semantics}
$\interp{[\cdot]^{\operatorname{ZH}}}=\interp{\cdot}$.
\end{propC}

\subsection{From $\cat{ZH}$ to $\cat{SOP}$}

Any $\cat{ZH}$-diagram can be understood as a $\cat{SOP}$-morphism. To do so, we use the PROP-functor $[\cdot]^{\operatorname{sop}}:\cat{ZH}\to\cat{SOP}$ defined as:
\[\left[~
\input{./figures/H-spider-phase.tikz}
~\right]^{\operatorname{sop}} := \sum e^{2i\pi \frac{\alpha}{2\pi}x_1\ldots x_ny_1\ldots y_m}\ketbra{y_1,\ldots, y_m}{x_1,\ldots,x_n}\]
\[\left[~
\begin{tikzpicture}
	\begin{pgfonlayer}{nodelayer}
		\node [style=box] (0) at (0, 0) {$s$};
	\end{pgfonlayer}
\end{tikzpicture}

~\right]^{\operatorname{sop}} := s\ketbra{}{} \quad\text{ for }s\in\mathbb R\]
\[\left[~
\input{./figures/Z-spider.tikz}
~\right]^{\operatorname{sop}} := \sum_y \ketbra{y,\ldots, y}{y,\ldots, y}\qquad\qquad
\left[~
\input{./figures/H-spider-0.tikz}
~\right]^{\operatorname{sop}} := \left[~
\input{./figures/H-spider-0-decomp.tikz}
~\right]^{\operatorname{sop}}\]
The functor furthermore maps the symmetric braiding (resp. the compact structure) of $\cat{ZH}$ to the symmetric braiding (resp. the compact structure) of $\cat{SOP}$.

This does not give a full description of $[\cdot]^{\operatorname{sop}}$, as we did not describe the interpretation of the H-spider for all parameters, but only for phases and $0$. However, any H-spider can be decomposed using the previous ones:

\begin{lem}
\label{lem:H-spider-decomp}
For any $r\in\mathbb C$ such that $\abs{r}\notin\{0,1\}$, there exist $s\in\mathbb C$, $\alpha,\beta\in\mathbb R$ such that:
\[\interp{
\input{./figures/H-spider.tikz}
} = \interp{
\input{./figures/H-spider-decomp.tikz}
}\]
\end{lem}

\begin{proof}
\phantomsection\label{prf:H-spider-decomp}
First, thanks to rule (HS1), we have $
\input{./figures/H-spider.tikz}
\eq{}
\input{./figures/H-spider-decomp-aux.tikz}
$. Then, we have:
\def\fig{H-spider-decomp-proof}
\begin{align*}
\interp{\input{./figures/\fig/\fig_00.tikz}}=\frac12\begin{pmatrix}1+r\\1-r\end{pmatrix}=\frac{1+r}2\begin{pmatrix}1\\\frac{1-r}{1+r}\end{pmatrix}
\end{align*}
and
\begin{align*}
\interp{\input{./figures/\fig/\fig_02.tikz}}
=2se^{i\frac\alpha2}\begin{pmatrix}\cos{\frac\alpha2}\\-ie^{i\beta}\sin{\frac\alpha2}\end{pmatrix}
=2se^{i\frac\alpha2}\cos{\frac\alpha2}\begin{pmatrix}1\\e^{i(\beta-\frac\pi2)}\tan{\frac\alpha2}\end{pmatrix}
\end{align*}
Hence, when $\abs{r}\notin\{0,1\}$, we have equality between the two with
$\alpha := 2\atan{\frac{1-r}{1+r}}$, $\beta = \arg\left(\frac{1-r}{1+r}\right)+\frac\pi2$ and $s:=\frac{1+r}{4e^{i\frac\alpha2}\cos{\frac\alpha2}}$ (since $r\neq1$, $\alpha$ is well defined and $\alpha\neq\pi\mod 2\pi$ so $s$ is also well-defined). From this, we get:
\[
\begin{array}[b]{c}
\interp{
\input{./figures/H-spider.tikz}
} = \interp{
\input{./figures/H-spider-decomp.tikz}
}
\\[-\baselineskip]
\end{array}
\qedhere
\]
\end{proof}

As a consequence, we extend the definition of $[\cdot]^{\operatorname{sop}}$ by:
\[\left[~
\input{./figures/H-spider.tikz}
~\right]^{\operatorname{sop}} :=\left[~
\input{./figures/H-spider-decomp.tikz}
~\right]^{\operatorname{sop}} \]

This interpretation of ZH-diagrams as $\cat{SOP}$-morphisms preserves the semantics:
\begin{propC}[\cite{SOP-Clifford}]
\label{prop:sop-preserves-semantics}
$\interp{[\cdot]^{\operatorname{sop}}} = \interp{\cdot}$. In other words, the following diagram commutes:
$$
\input{./figures/sop-interp-cd.tikz}
$$
\end{propC}

The composition of the two interpretations is the identity up to small rewrites:
\begin{propC}[\cite{SOP-Clifford}]
\label{prop:double-interp-is-identity}
$\left[[\cdot]^{\operatorname{ZH}}\right]^{\operatorname{sop}} \underset{\operatorname{TH}}\sim (\cdot)$
\end{propC}

\subsection{Restrictions of $\cat{SOP}$}

Recall that $\cat{ZH}_{\operatorname{TH}}$ exactly captures the Toffoli-Hadamard fragment of quantum mechanics. We can then use the two interpretations to define the Toffoli-Hadamard fragment of $\cat{SOP}$. We actually go a step beyond and define a family of fragments indexed by $n$:

\begin{defi}[$\cat{SOP}{[}\frac1{2^n}{]}$]
We define $\cat{SOP}[\frac1{2^n}]$ as the restriction of $\cat{SOP}$ to morphisms of the form: $\displaystyle t = \frac1{\sqrt2^p}\sum e^{2i\pi \frac P{2^n}}\ketbra{\vec O}{\vec I}$ where $p\in\mathbb Z$ and $P$ has integer coefficients.
\end{defi}

The Toffoli-Hadamard fragment is then the first such restriction ($n=1$):
\begin{prop}
\label{prop:sop1/2-TOF}
$\cat{SOP}[\frac12]$ captures exactly the Toffoli-Hadamard fragment of quantum mechanics.
\end{prop}
\begin{proof}
We can prove this by showing that $\left[\cat{ZH}_{\operatorname{TH}}\right]^{\operatorname{sop}}\subseteq \cat{SOP}[\frac12]$ and that $\left[\cat{SOP}[\frac12]\right]^{\operatorname{ZH}}\subseteq \cat{ZH}_{\operatorname{TH}}$. The two claims are straightforward verifications, and use the fact that compositions of $\cat{SOP}[\frac12]$-morphisms give $\cat{SOP}[\frac12]$-morphisms.\\
Then, $\interp{\cat{ZH}_{\operatorname{TH}}}=\interp{\left[\cat{ZH}_{\operatorname{TH}}\right]^{\operatorname{sop}}}\subseteq\interp{\cat{SOP}[\frac12]} = \interp{\left[\cat{SOP}[\frac12]\right]^{\operatorname{ZH}}}\subseteq \interp{\cat{ZH}_{\operatorname{TH}}}$, so:
\[\interp{\cat{SOP}[\textstyle\frac12]} = \interp{\cat{ZH}_{\operatorname{TH}}}\qedhere\]
\end{proof}
Notice in particular that the Hadamard and Toffoli gates given in \autoref{ex:H-Tof} lie in this fragment; and that compositions of elements of this fragment remain in the fragment. Not all of $\cat{SOP}[\frac12]$ can be generated by the two gates $H$ and $\operatorname{Tof}$, however, as $\cat{SOP}[\frac12]$ comprises linear maps that are not unitary, i.e.~such that $\interp{t^\dagger\circ t}\neq id$.

\subsection{The Rewrite System $\underset{\operatorname{TH}}\longrightarrow$, Graphically}
\label{sec:rewrites-in-ZH}

Before moving on to show the result of completeness the Toffoli-Hadamard fragment of $\cat{SOP}$, it is worth checking how the rules of $\rewrite{\operatorname{TH}}$ translate in $\cat{ZH}$. We will focus on the rules mentioned above, that were not present in the previous works on $\cat{SOP}$, namely (HHgen), (HHgen'), (HHnl) and (Rem).

\def\fig{HHgen}
Rule (HHgen) uses a side condition, that is $QQ'=Q$, which translates to
\begin{align}
\label{eq:QQp=Qp}
\input{./figures/\fig/\fig_00.tikz}
\eq{}\input{./figures/\fig/\fig_01.tikz}.
\end{align}
The pattern $\frac{y_0}2(y_i\widehat{Q}+\widehat{Q}'+1)$ is represented by:
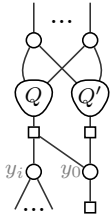
\begin{align*}
\input{./figures/\fig/\fig_02.tikz}
\end{align*}
and can be rewritten into:
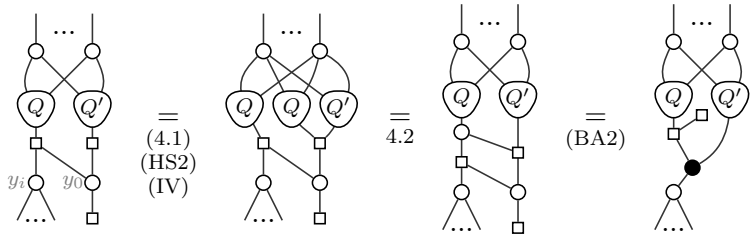
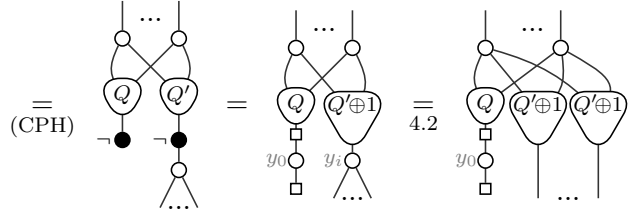
\begin{align*}
\input{./figures/\fig/\fig_02.tikz}
&\eq{\eqref{eq:QQp=Qp}\\\text{(HS2)}\\\text{(IV)}}\input{./figures/\fig/\fig_03.tikz}
\eq{\ref{lem:poly-copy}}\input{./figures/\fig/\fig_04.tikz}
\eq{\text{(BA2)}}\input{./figures/\fig/\fig_05.tikz}\\
&\eq{\text{(CPH)}}\input{./figures/\fig/\fig_06.tikz}
\eq{}\input{./figures/\fig/\fig_07.tikz}
\eq{\ref{lem:poly-copy}}\input{./figures/\fig/\fig_08.tikz}
\end{align*}
where the last equality (that uses \autoref{lem:poly-copy}) is the substitution $[y_i\leftarrow Q'\oplus1]$.

Rule (HHgen') can be proven graphically exactly in the same way, except we start from the second diagram in the derivation above.

Rule (HHnl) turns an occurrence of $\frac{y_0}2\widehat Q+\frac{y_0'}2\widehat Q'$ into $\frac{y_0}2(\widehat Q+\widehat Q'+\widehat Q\widehat Q')$, when the two variables are linked to nothing else than their respective polynomials $Q$ and $Q'$. The induced $\cat{ZH}$ identity can be derived using its rules:
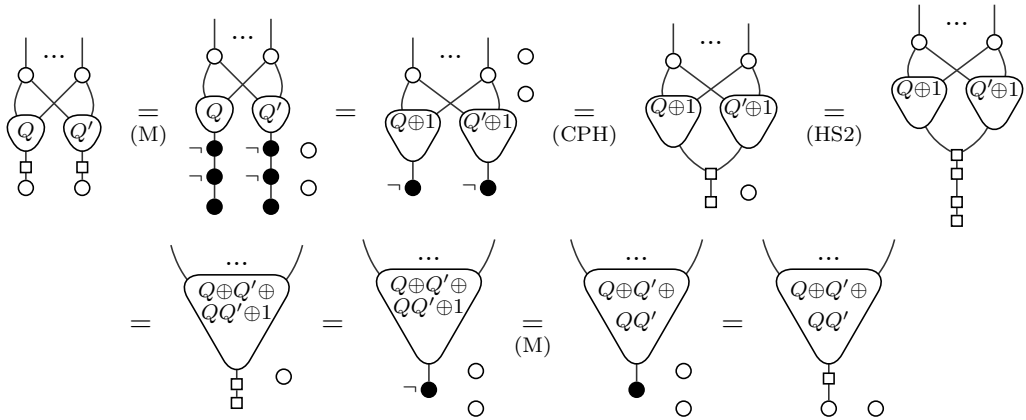
\def\fig{HHnl}
\begin{align*}
\input{./figures/\fig/\fig_00.tikz}
&\eq{\text{(M)}}\input{./figures/\fig/\fig_01.tikz}
\eq{}\input{./figures/\fig/\fig_02.tikz}
\eq{\text{(CPH)}}\input{./figures/\fig/\fig_03.tikz}
\eq{\text{(HS2)}}\input{./figures/\fig/\fig_04.tikz}\\
&\eq{}\input{./figures/\fig/\fig_05.tikz}
\eq{}\input{./figures/\fig/\fig_06.tikz}
\eq{\text{(M)}}\input{./figures/\fig/\fig_07.tikz}
\eq{}\input{./figures/\fig/\fig_08.tikz}
\end{align*}
Although the overall number of nodes usually increases, the number of white nodes that amount to $\cat{SOP}$-variables (i.e. white nodes that are not part of a polynomial) decreases.

Finally, Rule (Rem) turns the phase polynomial $\frac{y_0}2\widehat{Q}+\widehat{SQ}+\widehat{R}$ into $\frac{y_0}2\widehat{Q}+\widehat{R}$, which can be derived graphically as:
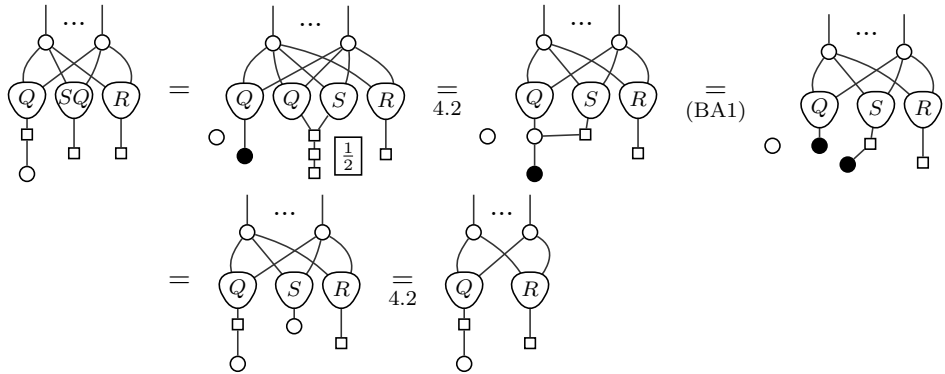
\def\fig{Rem}
\begin{align*}
\input{./figures/\fig/\fig_00.tikz}
&\eq{}\input{./figures/\fig/\fig_01.tikz}
\eq{\ref{lem:poly-copy}}\input{./figures/\fig/\fig_02.tikz}
\eq{\text{(BA1)}}\input{./figures/\fig/\fig_03.tikz}\\
&\eq{}\input{./figures/\fig/\fig_04.tikz}
\eq{\ref{lem:poly-copy}}\input{./figures/\fig/\fig_05.tikz}
\end{align*}

\section{Completeness for Toffoli-Hadamard}
\label{sec:TH-completeness}

In this section, we aim to show that the set of rules $\rewrite{\text{TH}}$ captures the whole Toffoli-Hadamard fragment of quantum mechanics. We do so by transporting the similar result from $\cat{ZH}_{\operatorname{TH}}$ to $\cat{SOP}[\frac12]$. First, we show:

\begin{prop}
\label{prop:TH-proves-TH}
$\forall D_1,D_2\in \cat{ZH}_{\operatorname{TH}},~~ \operatorname{ZH}_{\operatorname{TH}}\vdash D_1=D_2\implies [D_1]^{\operatorname{sop}}\underset{\textnormal{TH}}\sim [D_2]^{\operatorname{sop}}$
\end{prop}

\begin{proof}
\phantomsection\label{prf:TH-proves-TH}
We show that all the rules of $\operatorname{ZH}_{\operatorname{TH}}$ hold in $\cat{SOP}[\frac12]$, which together with \autoref{prop:TH-local} proves the result.

The translation of rule (IP) is implicit in $\cat{SOP}$: it essentially states that $y_i\cdot y_i = y_i$ for boolean variable $y_i$. 
Checking the rules (ZS1), (ZS2), (HS1), (HS2) and (M) is straightforward using the rule (HH). We give for instance a check of the rule (ZS1):
\begin{align*}
\left[
\input{./figures/ZH-rule-ZS1-lhs.tikz}
\right]^{\operatorname{sop}} &= \frac12\sum e^{2i\pi \frac{y_0+y_1}2y'}\ketbra{y_1,...,y_1}{y_0,...,y_0}\\
&\rewrite{\text{HH}(y',[y_1\leftarrow y_0])}
\sum \ketbra{y_0,...,y_0}{y_0,...,y_0} = \left[
\input{./figures/Z-spider.tikz}
\right]^{\operatorname{sop}}
\end{align*}
We give derivations to prove the remaining rules of $\operatorname{ZH}_{\operatorname{TH}}$. Recall that equality is up to $\alpha$-conversion.

(IV):
\begin{align*}
\left[~
\begin{tikzpicture}
	\begin{pgfonlayer}{nodelayer}
		\node [style=box] (0) at (0, 0) {$\frac12$};
	\end{pgfonlayer}
\end{tikzpicture}

~~
\begin{tikzpicture}
	\begin{pgfonlayer}{nodelayer}
		\node [style=white dot] (0) at (0, 0) {};
	\end{pgfonlayer}
\end{tikzpicture}

~\right]^{\operatorname{sop}} = \frac12\sum_y\ketbra{}{}
\rewrite{\text{Elim}} 1 = \left[~
\input{./figures/empty.tikz}
~\right]^{\operatorname{sop}}
\end{align*}

(Z):
\begin{align*}
\left[~
\begin{tikzpicture}
	\begin{pgfonlayer}{nodelayer}
		\node [style=box] (0) at (0, 0) {$\frac1{\sqrt2}$};
	\end{pgfonlayer}
\end{tikzpicture}

~~
\begin{tikzpicture}
	\begin{pgfonlayer}{nodelayer}
		\node [style=box] (0) at (0, 0.25) {};
		\node [style=white dot] (1) at (0, -0.25) {};
	\end{pgfonlayer}
	\begin{pgfonlayer}{edgelayer}
		\draw (0) to (1);
	\end{pgfonlayer}
\end{tikzpicture}

~\right]^{\operatorname{sop}}
\rewrite{\text{HH}}
\frac1{\sqrt2}\sum e^{2i\pi \frac y2}\ketbra{}{}
\rewrite{\text Z}
\sum e^{2i\pi \frac y2}\ketbra{}{}
\underset{\text{HH}}{\longleftarrow}
\left[~
\begin{tikzpicture}
	\begin{pgfonlayer}{nodelayer}
		\node [style=box] (0) at (0, 0.25) {};
		\node [style=white dot] (1) at (0, -0.25) {};
	\end{pgfonlayer}
	\begin{pgfonlayer}{edgelayer}
		\draw (0) to (1);
	\end{pgfonlayer}
\end{tikzpicture}

~\right]^{\operatorname{sop}}
\end{align*}

The two rules (BA1) and (BA2) are fairly easy to check, once one realises that $\left[
\input{./figures/xor.tikz}
\right]^{\operatorname{sop}} \rewrite{\text{HH}} \sum \ketbra{y_0{\oplus}y_1}{y_0,y_1}$:
\begin{align*}
\left[
\input{./figures/ZH-rule-BA1-rhs.tikz}
\right]^{\operatorname{sop}} \rewrite{\text{HH}}
\frac12\sum e^{2i\pi \frac{y_1+...+y_n+y_0}2y'} \ketbra{y_1,...,y_n}{y_0,...,y_0}\\
\rewrite{\text{HH}(y',[y_0\leftarrow y_1{\oplus}...{\oplus}y_n])}
\sum \ketbra{y_1,...,y_n}{y_1{\oplus}...{\oplus}y_n,...,y_1{\oplus}...{\oplus}y_n}
\underset{\text{HH}}{\longleftarrow}\left[
\input{./figures/ZH-rule-BA1-lhs.tikz}
\right]^{\operatorname{sop}} 
\end{align*}
and
\begin{align*}
\left[
\input{./figures/ZH-rule-BA2-rhs.tikz}
\right]^{\operatorname{sop}} \rewrite{\text{HH}}
\frac12\sum e^{2i\pi \left(\frac{y_1...y_my'}2+\frac{x_1+...+x_n+y'}2y''\right)} \ketbra{y_1,...,y_m}{x_1,...x_n}\\
\rewrite{\text{HH}(y'',[y'\leftarrow x_1{\oplus}...{\oplus}x_n])}
\sum e^{2i\pi \left(\frac{y_1...y_mx_1}2+...+\frac{y_1...y_mx_n}2\right)} \ketbra{y_1,...,y_m}{x_1,...x_n}\\
\underset{\text{HH}}{\longleftarrow}\left[
\input{./figures/ZH-rule-BA2-lhs.tikz}
\right]^{\operatorname{sop}} 
\end{align*}

(CPH):
\begin{align*}
\left[
\input{./figures/ZH-rule-CPH-lhs.tikz}
\right]^{\operatorname{sop}}&\rewrite{\text{HH}} \sum e^{2i\pi\left(\frac{y_0}2+\frac{y_0y_1y_2}2\right)}\ket{y_1,y_2}
\rewrite{\text{HHgen}(y_0,[y_1\leftarrow 1])}
\sum e^{2i\pi\left(\frac{y_0}2+\frac{y_0y_2}2\right)}\ket{1,y_2}\\
&\rewrite{\text{HH}(y_0,[y_2\leftarrow 1])} 2\ket{1,1}
\underset{\text{HH}}{\longleftarrow}\left[
\input{./figures/ZH-rule-CPH-rhs.tikz}
\right]^{\operatorname{sop}} \qedhere
\end{align*}
\end{proof}

\begin{rem}
Notice in the previous proof that Rule (CPH) is the only one that requires (HHgen) to be proven. It is unclear as of now if (CPH) is a necessary axiom of $\operatorname{ZH}_{\operatorname{TH}}$, but unsuccessful efforts put into trying to derive both (CPH) in $\operatorname{ZH}_{\operatorname{TH}}$ and (HHgen) in $\underset{\operatorname{TH}}\sim$ hint towards their respective necessity.
\end{rem}

We can now use the previous proposition to show the main result of this paper:
\begin{thm}
\label{thm:TH-completeness}
$\cat{SOP}[\frac12]/\underset{\textnormal{TH}}\sim$ is complete, i.e.: 
$\forall t_1,t_2\in\cat{SOP}[{\textstyle\frac12}],~~\interp{t_1}=\interp{t_2}\iff t_1\underset{\textnormal{TH}}\sim t_2$
\end{thm}

\begin{proof}
Let $t_1$ and $t_2$ be two $\cat{SOP}[\frac12]$-morphisms such that $\interp{t_1}=\interp{t_2}$. By \autoref{prop:zh-preserves-semantics}:
$\interp{[t_1]^{\operatorname{ZH}}} = \interp{[t_2]^{\operatorname{ZH}}}$.\\
By completeness of $\cat{ZH}_{\operatorname{TH}}/\operatorname{ZH}_{\operatorname{TH}}$ (\autoref{thm:ZH-completeness}):
$\operatorname{ZH}_{\operatorname{TH}}\vdash [t_1]^{\operatorname{ZH}}=[t_2]^{\operatorname{ZH}}$\\
Thanks to \autoref{prop:TH-proves-TH}:
$\left[[t_1]^{\operatorname{ZH}}\right]^{\operatorname{sop}} \underset{\text{TH}}\sim \left[[t_2]^{\operatorname{ZH}}\right]^{\operatorname{sop}}$.
Finally, by \autoref{prop:double-interp-is-identity}:
\[t_1\underset{\text{TH}}\sim\left[[t_1]^{\operatorname{ZH}}\right]^{\operatorname{sop}} \underset{\text{TH}}\sim \left[[t_2]^{\operatorname{ZH}}\right]^{\operatorname{sop}}\underset{\text{TH}}\sim t_2\qedhere\]
\end{proof}

The rewrite system is however not sufficient to get to a unique normal form, as:
\begin{lem}[Non-Confluence]
The rewrite system $\rewrite{\textnormal{TH}}$ is not confluent.
\end{lem}

\begin{proof}
We can exhibit a term which can be reduced into two different irreducible terms. To do so, we can use an instance of derived Rule (HHnl). As a by-product, this shows that rule (HHnl) would be necessary in the oriented setting, i.e.~in $\underset{\textnormal{TH}}\longrightarrow$ rather than in $\underset{\textnormal{TH}}\sim$.
The  $\cat{SOP}[\frac12]$-morphism:
\[\sum e^{2i\pi \left(\frac{y_0}2 + \frac{y_0y_1(y_2y_3+1)}2 + \frac{y_1y_0'}2 + \frac{y_0'(y_4y_5+1)}2 \right)}\ket{y_2y_3y_4y_5}
\]
can be reduced into (at least) two different irreducible terms:
\begin{gather*}
\sum e^{2i\pi \left(\frac{y_0}2 + \frac{y_0y_1(y_2y_3+1)}2 + \frac{y_1y_0'}2 + \frac{y_0'(y_4y_5+1)}2 \right)}\ket{y_2y_3y_4y_5}\\
\begin{array}[b]{cc}
\phantom{\text{HHgen}(y_1\leftarrow1)}\Big\downarrow\text{HHgen}(y_1\leftarrow1)&\phantom{\text{HH}(y_1\leftarrow y_4y_5\oplus1)}\Big\downarrow\text{HH}(y_1\leftarrow y_4y_5\oplus1)\\[0.7em]
\sum e^{2i\pi \left(\frac{y_0}2 y_2y_3+ \frac{y_0'}2 y_4y_5\right)}\ket{y_2y_3y_4y_5} &\qquad
2\sum e^{2i\pi \left(\frac{y_0}2(y_2y_3 + y_4y_5 + y_2y_3y_4y_5)\right)}\ket{y_2y_3y_4y_5}
\end{array}\qedhere
\end{gather*}
\end{proof}
Another important downside is the potential explosion of the size of the phase polynomial:
\begin{lem}
Applying (HH) $k$ times in a row on an SOP morphism with phase polynomial of size $O(k)$ may give a morphism with phase polynomial of size $O(2^k)$.
\end{lem}
\begin{proof}
For any $k\geq1$ we can define the following term:
\[t_k:=\sum e^{2i\pi \left(y_1\cdot...\cdot y_k+\sum\limits_{i=1}^k\frac{y'_{i}}2(y_i+x_i+x_i')\right)} \]
on which we can apply (HH) $k$ times in a row with $[y_i\leftarrow x_i\oplus x_i']$. In that case we end up with:
\[t_k\rewrite{}^k2^k\sum e^{2i\pi \left(\prod\limits_{i=1}^k(x_i+x_i')\right)}\]
While $t_k$ has only $3(k+1)$ terms (each of degree $2$ except one of degree $k$) in its phase polynomial, it can rewrite into a morphism with $2^{k+1}$ terms (each of degree $k$).
\end{proof}
Hence, if one were to perform simplifications with this rewrite system, they ought to give special attention as to where and in which order to apply the rules.

\section{Completeness for the Dyadic Fragment}
\label{sec:completeness}

We show here how we can turn an $\cat{SOP}[\frac1{2^{n+1}}]$-morphism into an $\cat{SOP}[\frac1{2^{n}}]$-morphism in a ``reversible'' manner. This will allow us to extend the completeness result to all the restrictions $\cat{SOP}[\frac1{2^n}]$. This is particularly interesting as the phase gates with dyadic multiples of $\pi$, used in particular in the quantum Fourier transform, belong in these fragments:
\[R_Z\left(p\frac\pi{2^k}\right) := \sum_{y_0}e^{2i\pi\cdot\frac p{2^{k-1}}}\ketbra{y_0}{y_0}\]

\begin{rem}
The construction we are about to define to relate the different dyadic levels together can be seen as a (non-unitary) instance of ``catalysis'', defined and studied at length in the recent paper \cite{Amy2023catalytic}, where the complexity of applying gate (or piece of circuit) $U$ several times in the overall circuit is offset to a particular state $\ket{\chi_U}$ and several occurrences of a ``simpler'' piece of circuit $\phi_U$ such that $\phi_U\circ(id\otimes \ket{\chi_U}) = U\otimes \ket{\chi_U}$. In our case, $U$ is an occurrence of the $\frac\pi{2^{n+1}}$ phase, and $\ket{\chi_U}$ encodes that phase. Both $U$ and $\ket{\chi_U}$ live in $\cat{SOP}[\frac1{2^{n+1}}]$, while $\phi_U$ lives in $\cat{SOP}[\frac1{2^{n}}]$.
\end{rem}

\subsection{Ascending the Dyadic Levels}

These transformations between restrictions of $\cat{SOP}$ are more easily defined on $\cat{SOP}$-morphisms of a particular shape, namely, when their phase polynomial is reduced to a single monomial. Because of this, we show how a $\cat{SOP}$-morphism can be turned into a composition of these.
\begin{lem}
\label{lem:sop-decomp}
Let $P = \sum m_i \in \mathbb R[X_1,\ldots,X_{k}]/(X_i^2-X_i)$, and $t = s\sum e^{2i\pi P} \ketbra{\vec O}{\vec I}$. Then:
\begin{align*}
\left[\begin{array}{c}
\displaystyle\left(s\sum \ketbra{\vec O}{y_0,...,y_k}\right)\circ\hfill\phantom{.}\\
\displaystyle\left(\sum e^{2i\pi m_1}\ketbra{y_0,...,y_k}{y_0,...,y_k}\right)
\circ \ldots \circ
\left(\sum e^{2i\pi m_\ell}\ketbra{y_0,...,y_k}{y_0,...,y_k}\right)\\
\displaystyle\hfill\circ\left(\sum \ketbra{y_0,...,y_k}{\vec I}\right)
\end{array}\right]
\overset\ast{\rewrite{\textnormal{HH}}}t
\end{align*}
\end{lem}

\begin{proof}
Let us start by composing the two first diagonal terms:
\begin{align*}
\left(\sum e^{2i\pi m_1}\ketbra{y_0,...,y_k}{y_0,...,y_k}\right)
\circ
\left(\sum e^{2i\pi m_\ell}\ketbra{y_0,...,y_k}{y_0,...,y_k}\right)\\
= \frac1{2^{k+1}}\sum e^{2i\pi (m_1 + m_2[y_i\leftarrow y_i'] + \frac{y_0y_0'' + y_0''y_0'}2+...+\frac{y_ky_k'' + y_k''y_k'}2)}\ketbra{y_0,...,y_k}{y_0',...,y_k'}\\
\rewrite{\text{HH}(y_i'\leftarrow y_i)}
\sum e^{2i\pi (m_1 + m_2)}\ketbra{y_0,...,y_k}{y_0,...,y_k}
\end{align*}
Doing so repeatedly with all the diagonal terms gives
\[\sum e^{2i\pi P}\ketbra{y_0,...,y_k}{y_0,...,y_k}\]
Finally, applying the extremal terms (one at a time) and removing the newly created variables with (HH), just as we did for the diagonal terms, yields $t$.
\end{proof}

Notice that this decomposed form is not unique, as different orderings on the monomials of $P$ define different orderings of the compositions. However, this will not matter.

A particular care is sadly needed for the overall scalar. Because of this, we will first focus on a slightly different notion of restriction of $\cat{SOP}$.

\begin{defi}[$\cat{SOP}{[}\frac1{2^n}{]}'$]
We define $\cat{SOP}[\frac1{2^n}]'$ as the restriction of $\cat{SOP}$ to morphisms of the form: $\displaystyle t = \frac1{2^p}\sum e^{2i\pi \frac P{2^n}}\ketbra{\vec O}{\vec I}$ where $P$ has integer coefficients.
\end{defi}

The only difference with $\cat{SOP}[\frac1{2^n}]$ is that the overall scalar is now a power of $\frac12$ and not of $\frac1{\sqrt2}$. There always exists a $\cat{SOP}[\frac1{2^n}]'$-morphism that represents the same linear map as any $\cat{SOP}[\frac1{2^{n}}]$-morphism.

\begin{lem}
$\interp{\frac1{\sqrt2}\sum\limits_{y_0\in V} e^{2i\pi\left(\frac{1}{8} + \frac{3}{4}y_{0}\right)}} = 1$. Hence:
\[\textstyle\forall t\in \cat{SOP}[\frac1{2^n}],~\exists t'\in \cat{SOP}[\frac1{2^{\max(3,n)}}]',~~\interp{t}=\interp{t'}\]
\end{lem}
\begin{proof}
If $t\in \cat{SOP}[\frac1{2^n}]$ and $t\notin \cat{SOP}[\frac1{2^n}]'$, then:
\[t':=t\otimes\left(\frac1{\sqrt2}\sum e^{2i\pi\left(\frac{1}{8} + \frac{3}{4}y_{0}\right)}\right) \in \cat{SOP}[\frac1{2^{\max(3,n)}}]'\quad \text{ and }\quad\interp{t'}=\interp{t}.\qedhere\]
\end{proof}

We can now define the family of maps that will link the different levels of the ``dyadic levels'':
\begin{defi}
For any $k\geq1$, we define the functor $\ascend{\cdot}_k:\cat{SOP}[\frac1{2^{k+1}}]'\to\cat{SOP}[\frac1{2^k}]'$, first for morphisms $t=s\sum e^{2i\pi \frac\ell{2^{k+1}}y_{i_1}...y_{i_q}} \ketbra{\vec O}{\vec I}$ with phase polynomial of size 0 or 1:
\[t\mapsto
\begin{cases}
\displaystyle s\sum e^{2i\pi \frac{\ell/2}{2^{k}}y_{i_1}...y_{i_q}} \ketbra{\vec O,y'}{\vec I,y'} = t\otimes id&\text{ if } \ell\bmod2 = 0\\[0.5ex]
\displaystyle s\sum e^{2i\pi \frac{y_{i_1}...y_{i_q}}{2^{k}}\left((\ell-1)/2+y'\right)} \ketbra{\vec O,y'}{\vec I,y'{\oplus}y_{i_1}...y_{i_q}}&\text{ if } \ell\bmod2 = 1
\end{cases}\]
The functor is then extended to any $\cat{SOP}[\frac1{2^{k+1}}]'$-morphism by the decomposition of \autoref{lem:sop-decomp} (and given a particular ordering on the monomials of the phase polynomial).
\end{defi}
Since $\ascend{\cdot}_k$ is defined to be a functor, we have $\ascend{\cdot\,\circ\,\cdot}_k = \ascend{\cdot}_k\circ\ascend{\cdot}_k$. We can show that the ordering of the monomials has no real importance. Indeed, suppose $t_1=\sum e^{2i\pi \frac{\ell_1}{2^{k+1}}y_{i_1}...y_{i_q}} \ketbra{\vec y}{\vec y}$ and $t_2=\sum e^{2i\pi \frac{\ell_2}{2^{k+1}}y_{j_1}...y_{j_r}} \ketbra{\vec y}{\vec y}$. Then: $
\ascend{t_1\circ t_2}_k = \ascend{t_2\circ t_1}_k$ 
quite obviously when either $\ell_1\bmod2=0$ or $\ell_2\bmod2=0$, but also when $\ell_1\bmod2=\ell_2\bmod2=1$:
\begin{align*}
\ascend{t_1\circ t_2}_k
\rewrite{\text{HH}} \sum e^{2i\pi \left(
\begin{array}{c}
\scriptstyle \frac{y_{i_1}...y_{i_q}}{2^{k}}\left((\ell_1-1)/2+y'\right)
 +\frac{y_{j_1}...y_{j_r}}{2^{k}}\left((\ell_2-1)/2+y'\right)\\[1ex]
\scriptstyle + \frac{y_{i_1}...y_{i_q}y_{j_1}...y_{j_r}}{2^{k}}(1-2y')
\end{array}
\right)} \raisebox{-2ex}{\hspace*{-4em}$\ketbra{\vec y,y'}{\vec y,y'{\oplus}y_{i_1}...y_{i_q}{\oplus}y_{j_1}...y_{j_r}}$}\\
\underset{\text{HH}}\longleftarrow \ascend{t_2\circ t_1}_k
\end{align*}
Notice however that $\ascend{\cdot}_k$ adds an input and an output, so necessarily $\ascend{\cdot\otimes\cdot}_k\neq\ascend\cdot_k\otimes\ascend\cdot_k$.

The functors $\ascend{\cdot}_k$ map terms with the same semantics to terms with the same semantics:
\begin{prop}
\label{prop:ascending-encodes}
$\displaystyle\forall t_1,t_2\in\cat{SOP}[\textstyle\frac1{2^{k+1}}]',~~\interp{t_1}=\interp{t_2}\implies \interp{\ascend{t_1}_k} = \interp{\ascend{t_2}_k}$
\end{prop}
\begin{proof}
\phantomsection\label{prf:ascending-encodes}
We prove this proposition by showing that:
\begin{enumerate}
\item $\interp{\cat{SOP}[\frac1{2^{k+1}}]'}\subseteq
\mathcal M(\mathbb Z[\frac12,e^{i\frac\pi{2^k}}])$
\label{item:interp}
\item For each element $x\in \mathbb Z[\frac12,e^{i\frac\pi{2^k}}]$, there exists a unique decomposition as $x = x_1+e^{i\frac\pi{2^k}}x_2$ where $x_1,x_2\in \mathbb Z[\frac12,e^{i\frac\pi{2^{k-1}}}]$
\label{item:unique}
\item There exists a map $\psi_k:\mathcal M(\mathbb Z[\frac12, e^{i\frac\pi{2^k}}])\to \mathcal M(\mathbb Z[\frac12, e^{i\frac\pi{2^{k-1}}}])$, based on the decomposition, and such that $\interp{\ascend{t}_k} = \psi_k\left(\interp{t}\right)$
\label{item:exists}
\end{enumerate}
In this case, given $t_1,t_2\in\cat{SOP}[\frac1{2^{k+1}}]'$ such that $\interp{t_1}=\interp{t_2}$, by (\ref{item:interp}) we can apply $\psi_k$ to their interpretation. By uniqueness of the decomposition (\ref{item:unique}), $\psi_k(\interp{t_1})=\psi_k(\interp{t_2})$. Finally, by (\ref{item:exists}), $\interp{\ascend{t_1}_k} = \interp{\ascend{t_2}_k}$.
Let us now prove the previous claims:
\begin{enumerate}
\item This point is a simple verification.
\item Let $\displaystyle x=\sum_{\ell=0}^{2^k-1} \alpha_\ell e^{i\frac{\ell\pi}{2^k}}$ $\in \mathbb Z[\frac12,e^{i\frac\pi{2^k}}]$. Obviously, $x$ can be decomposed as $$ x = \sum_{\ell=0}^{2^{k-1}-1} \alpha_{2\ell} e^{i\frac{\ell\pi}{2^{k-1}}} + e^{i\frac{\pi}{2^k}}\sum_{\ell=0}^{2^{k-1}-1} \alpha_{2\ell+1} e^{i\frac{\ell\pi}{2^{k-1}}} = x_1+e^{i\frac\pi{2^k}}x_2$$
where $x_1,x_2\in \mathbb Z[\frac12,e^{i\frac\pi{2^{k-1}}}]$. We now need to show that this decomposition is unique. To do so, let us consider $\mathbb Q[e^{i\frac\pi{2^k}}]$ and $\mathbb Q[e^{i\frac\pi{2^{k-1}}}]$. These are two fields such that $\mathbb Q[e^{i\frac\pi{2^{k-1}}}] \subset \mathbb Q[e^{i\frac\pi{2^k}}]$. $\mathbb Q[e^{i\frac\pi{2^k}}]$ can hence be seen as a vector space over $\mathbb Q[e^{i\frac\pi{2^{k-1}}}]$. This vector space is of dimension:
$$\left[\mathbb Q[e^{i\frac\pi{2^k}}]:\mathbb Q[e^{i\frac\pi{2^{k-1}}}]\right]
= \left[\mathbb Q[e^{i\frac{2\pi}{2^{k+1}}}]:\mathbb Q[e^{i\frac{2\pi}{2^k}}]\right]
=\frac{\left[\mathbb Q[e^{i\frac{2\pi}{2^{k+1}}}]:\mathbb Q\right]}{\left[\mathbb Q[e^{i\frac{2\pi}{2^k}}]:\mathbb Q\right]} = \frac{\varphi(2^{k+1})}{\varphi(2^k)}=\frac{2^k}{2^{k-1}}=2$$
where $\varphi$ is Euler's totient function. The vector space has $(1,e^{i\frac\pi{2^k}})$ as a basis. Hence, the above decomposition is unique.
\item We now need to define $\psi_k$. We are going to define it first on scalars, and on the basis $(1,e^{i\frac\pi{2^k}})$:
$$\psi_k(1):= I_2 = \begin{pmatrix}1&0\\0&1\end{pmatrix}\qquad\quad\text{and}\qquad\quad
\psi_k(e^{i\frac\pi{2^k}}):= X_k = \begin{pmatrix}0&1\\e^{i\frac\pi{2^{k-1}}}&0\end{pmatrix}$$
By linearity, $\psi_k$ is defined on all elements of $\mathbb Z[\frac12,e^{i\frac\pi{2^k}}]$. We then naturally extend this definition to any matrix over these elements. Formally: $\psi_k:A+Be^{i\frac\pi{2^k}}\mapsto A\otimes I_2 + B\otimes X_k$ where $A+Be^{i\frac\pi{2^k}}$ is the aforementioned decomposition extended to matrices. One can check that $\psi_k$ is a homomorphism, i.e.~$\psi_k(.+.)=\psi_k(.)+\psi_k(.)$ and $\psi_k(.\circ .)=\psi_k(.)\circ\psi_k(.)$.

It remains to show that $\interp{\ascend{.}_k} = \psi_k\left(\interp{.}\right)$. Since $\psi_k$ is a homomorphism, it is enough to show the result on the terms in the decomposed form of \autoref{lem:sop-decomp}. Let $t=s\sum e^{2i\pi \frac\ell{2^{k+1}}y_{i_1}...y_{i_q}} \ketbra{\vec O}{\vec I}$ be such a term.

If $\ell\bmod2=0$, then $\interp{t}\in\mathcal M(\mathbb Z[\frac12, e^{i\frac\pi{2^{k-1}}}])$ so $\psi_k(\interp t) = \interp t \otimes I_2$ and: \[\interp{\ascend{t}_k} = \interp{s\sum e^{2i\pi \frac{\ell/2}{2^{k}}y_{i_1}...y_{i_q}} \ketbra{\vec O,y'}{\vec I,y'}} = \interp t\otimes I_2.\]

If $\ell\bmod2=1$, then:
$$\interp{t} = se^{i\frac\pi{2^k}}\sum_{y_{i_1}...y_{i_q}=1} e^{2i\pi \frac{(\ell-1)/2}{2^k}} \ketbra{\vec O}{\vec I} + s\sum_{y_{i_1}...y_{i_q}=0} \ketbra{\vec O}{\vec I}$$
so:
$$\psi_k(\interp t) = \left(s\sum_{y_{i_1}...y_{i_q}=1} e^{2i\pi \frac{(\ell-1)/2}{2^k}} \ketbra{\vec O}{\vec I}\right)\otimes X_k + \left(s\sum_{y_{i_1}...y_{i_q}=0} \ketbra{\vec O}{\vec I}\right)\otimes I_2$$
and 
\begin{align*}
&\interp{\ascend{t}_k} = s\sum e^{2i\pi \frac{y_{i_1}...y_{i_q}}{2^{k}}\left((\ell-1)/2+y'\right)} \ketbra{\vec O,y'}{\vec I,y'{\oplus}y_{i_1}...y_{i_q}}\\
&= s\sum_{y_{i_1}...y_{i_q}=1} e^{2i\pi \frac{(\ell-1)/2+y'}{2^{k}}} \ketbra{\vec O,y'}{\vec I,y'{\oplus}1} + s\sum_{y_{i_1}...y_{i_q}=0} \ketbra{\vec O,y'}{\vec I,y'}\\
&= \left(s\sum_{y_{i_1}...y_{i_q}=1} e^{2i\pi \frac{(\ell-1)/2}{2^k}} \ketbra{\vec O}{\vec I}\right)\otimes X_k + \left(s\sum_{y_{i_1}...y_{i_q}=0} \ketbra{\vec O}{\vec I}\right)\otimes I_2 = \psi_k(\interp t)\qedhere
\end{align*}
\end{enumerate}
\end{proof}

\subsection{Going Back}

We  now show how to reverse the functors $\ascend{\cdot}_k$.

\begin{defi}
For any $k\geq1$, we define the (partial) map $\descend{\cdot}_k:\cat{SOP}[\frac1{2^k}]'\to \cat{SOP}[\frac1{2^{k+1}}]'$ as:
\[\forall t:n+1\to m+1 \in\cat{SOP}[{\textstyle\frac1{2^k}}]',~~\descend{t}_k:= (id_m\otimes \bra0)\circ t \circ (id_n\otimes \sum e^{2i\pi \frac{y_0}{2^{k+1}}}\ket{y_0})\]
\end{defi}
Notice that $\descend{\cdot}_k$ can only be applied on morphisms that have at least one input and one output.

$\descend{\cdot}_k$ reverses the action of $\ascend{\cdot}_k$ (up to some rewrites):
\begin{prop}
\label{prop:descending-reverses-ascending}
$\descend{\ascend{\cdot}_k}_k \underset{\textnormal{TH}}\sim (\cdot)$ and $t_1\underset{\textnormal{TH}}\sim t_2 \implies \descend{t_1}\underset{\textnormal{TH}}\sim \descend{t_2}$ for any two terms $t_1,t_2$.
\end{prop}

\begin{proof}
\phantomsection\label{prf:descending-reverses-ascending}
Again, we can use the decomposition given in \autoref{lem:sop-decomp}. We can show that if $t=s\sum e^{2i\pi \frac\ell{2^{k+1}}y_{i_1}...y_{i_q}} \ketbra{\vec O}{\vec I}$, then $\ascend{t}_k\circ (id_n\otimes \sum e^{2i\pi \frac{y_0}{2^{k+1}}}\ket{y_0}) \underset{\textnormal{TH}}\sim t\otimes \sum e^{2i\pi \frac{y_0}{2^{k+1}}}\ket{y_0}$:

If $\ell\bmod2=0$, then $\ascend{t}_k = t\otimes id$ so $\ascend{t}_k\circ (id_n\otimes \sum e^{2i\pi \frac{y_0}{2^{k+1}}}\ket{y_0}) \underset{\textnormal{TH}}\sim t\otimes \sum e^{2i\pi \frac{y_0}{2^{k+1}}}\ket{y_0}$.

If $\ell\bmod2=1$, then:
\begin{align*}
\ascend{t}_k\circ (id_n\otimes \sum e^{2i\pi \frac{y_0}{2^{k+1}}}\ket{y_0})
\hspace*{25em}\\
= \frac s2\sum e^{2i\pi \left( \frac{y_{i_1}...y_{i_q}}{2^{k}}\left((\ell-1)/2+y'\right)+\frac{y'+y_{i_1}...y_{i_q}+y_0}2y''+\frac{y_0}{2^{k+1}}\right)} \ketbra{\vec O,y'}{\vec I}\\
\rewrite{\text{HH}(y'',[y_0\leftarrow y'{\oplus}y_{i_1}...y_{i_q}])}
s\sum e^{2i\pi \left( \frac{y_{i_1}...y_{i_q}}{2^{k}}\left((\ell-1)/2+y'\right)+\frac{y'+y_{i_1}...y_{i_q}-2y'y_{i_1}...y_{i_q}}{2^{k+1}}\right)} \ketbra{\vec O,y'}{\vec I}\\
=s\sum e^{2i\pi \left( \ell\frac{y_{i_1}...y_{i_q}}{2^{k}}+\frac{y'}{2^{k+1}}\right)} \ketbra{\vec O,y'}{\vec I} = t\otimes \sum e^{2i\pi \frac{y_0}{2^{k+1}}}\ket{y_0}
\end{align*}

Now, for an arbitrary $t\in \cat{SOP}[\frac1{2^{k+1}}]'$, we can do the above inductively on each term in its decomposition, resulting in $\ascend{t}_k\circ (id_n\otimes \sum e^{2i\pi \frac{y_0}{2^{k+1}}}\ket{y_0}) \underset{\textnormal{TH}}\sim t\otimes \sum e^{2i\pi \frac{y_0}{2^{k+1}}}\ket{y_0}$. Finally:
\begin{align*}
\descend{\ascend{t}_k}_k &= (id_m\otimes \bra0)\circ \ascend{t}_k \circ (id_n\otimes \sum e^{2i\pi \frac{y_0}{2^{k+1}}}\ket{y_0})\\
&\underset{\textnormal{TH}}\sim (id_m\otimes \bra0)\circ \left(t\otimes \sum e^{2i\pi \frac{y_0}{2^{k+1}}}\ket{y_0}\right)
\underset{\textnormal{TH}}\sim t
\end{align*}

The second result in the Proposition simply comes from the fact that $\descend{t_i}$ is built by composition from $t_i$, so Proposition \ref{prop:TH-local} gives the desired result.
\end{proof}

\subsection{Completeness}

We may now show completeness first for $\cat{SOP}[\frac1{2^{k+1}}]'$ and then tweak the equational theory to extend the result to $\cat{SOP}[\frac1{2^{k+1}}]$.

\begin{thm}[Completeness of $\cat{SOP}{[}\frac1{2^{k+1}}{]}'/\sim_{\raisebox{-1ex}{\hspace*{-2ex}\textnormal{\scriptsize TH}}}$]
\label{thm:dyadic-completeness-aux}
\[\forall t_1,t_2\in \cat{SOP}[{\textstyle\frac1{2^{k+1}}}]',~~\interp{t_1}=\interp{t_2}\iff t_1\underset{\textnormal{TH}}\sim t_2\]
\end{thm}

\begin{proof}
Let $t_1,t_2\in \cat{SOP}[{\textstyle\frac1{2^{k+1}}}]'$ such that $\interp{t_1}=\interp{t_2}$. By \autoref{prop:ascending-encodes}:
\[\interp{\ascend{...\ascend{t_1}_k...}_1}=\interp{\ascend{...\ascend{t_2}_k...}_1}\]
Since $\ascend{...\ascend{t_i}_k...}_1 \in \cat{SOP}[\frac12]'\subset \cat{SOP}[\frac12]$, by completeness of this fragment (\autoref{thm:TH-completeness}): 
\[\ascend{...\ascend{t_1}_k...}_1 \underset{\textnormal{TH}}\sim \ascend{...\ascend{t_2}_k...}_1 \]
Finally, by \autoref{prop:descending-reverses-ascending}:
\[t_1\underset{\textnormal{TH}}\sim\descend{...\descend{\ascend{...\ascend{t_1}_k...}_1}_1...}_k \underset{\textnormal{TH}}\sim \descend{...\descend{\ascend{...\ascend{t_2}_k...}_1}_1...}_k\underset{\textnormal{TH}}\sim t_2.\qedhere\]
\end{proof}

This is not entirely satisfactory, as we would like to relate any two morphisms of the same interpretation. However:

\begin{lem}
\label{lem:empty-intersection-sop-sop'}
If $t_1\in \cat{SOP}[{\textstyle\frac1{2^{k+1}}}]'$ and $t_2\in\cat{SOP}[{\textstyle\frac1{2^{k+1}}}]\setminus \cat{SOP}[{\textstyle\frac1{2^{k+1}}}]'$, then $t_1\underset{\textnormal{TH}}{\nsim}t_2$.
\end{lem}

\begin{proof}
There is no rule in $\rewrite{\textnormal{TH}}$ that changes the overall scalar from an odd power of $\frac1{\sqrt2}$ to an even one, or vice-versa.
\end{proof}

However, adding a single rule:
\[\sum_{\vec y} e^{2i\pi \left(\frac18+\frac34y_0+R\right)}\ketbra{\vec O}{\vec I} \rewrite{y_0\notin\Var(R,\vec O,\vec I)} \sqrt2\sum_{\vec y\setminus\{y_0\}} e^{2i\pi R}\ketbra{\vec O}{\vec I} \tag{$\sqrt2$}\]
fixes this caveat. This rule can also be recovered from the more general one:
\[\sum_{\vec y} e^{2i\pi\left(\frac{y_0}{4} + \frac{y_0}{2}\widehat{Q} + R\right)}\ketbra{\vec O}{\vec I}
\rewrite{y_0\notin\Var(Q,R,\vec O,\vec I)} \sqrt{2}\sum_{\vec y\setminus{\{y_0\}}} e^{2i\pi\left(\frac{1}{8}-\frac{1}{4}\widehat{Q} + R\right)}\ketbra{\vec O}{\vec I}\tag{$\omega$}\]
which was already used in \cite{SOP,LvdWK,SOP-Clifford} to deal with the Clifford fragment of quantum mechanics.

With this additional rule at hand, we can derive the general completeness theorem:
\begin{thm}[Completeness of $\cat{SOP}{[}\frac1{2^{k+1}}{]}/\sim_{\raisebox{-1ex}{\hspace*{-2ex}\textnormal{\scriptsize TH'}}}$]
Let us write $\rewrite{\textnormal{TH'}}~:=~\rewrite{\textnormal{TH}}+\{(\sqrt2)\}$. Then: 
$\forall t_1,t_2\in \cat{SOP}[{\textstyle\frac1{2^{k+1}}}],~~\interp{t_1}=\interp{t_2}\iff t_1\underset{\textnormal{TH'}}\sim t_2$
\end{thm}

\begin{proof}
Let $t_1,t_2\in \cat{SOP}[{\textstyle\frac1{2^{k+1}}}]$ such that $\interp{t_1}=\interp{t_2}$. Let us also write:
\[t_{\sqrt2}:=\frac{1}{\sqrt2} \sum e^{2i\pi\left(\frac{1}{8} + \frac{3}{4}y_{0}\right)}\]
We define $t_i'$ as:
$\qquad t_i':=\begin{cases}
t_i &\text{ if } t_i\in \cat{SOP}[{\textstyle\frac1{2^{k+1}}}]'\\
t_i\otimes t_{\sqrt2} &\text{ if } t_i\notin \cat{SOP}[{\textstyle\frac1{2^{k+1}}}]'
\end{cases}$.\\
It is easy to check that $t_i'\in \cat{SOP}[{\textstyle\frac1{2^{\max(3,k+1)}}}]'$ and that $t_i\underset{\textnormal{TH'}}\sim t_i'$. By \autoref{thm:dyadic-completeness-aux}:
\[t_1 \underset{\textnormal{TH'}}\sim t_1' \underset{\textnormal{TH'}}\sim t_2' \underset{\textnormal{TH'}}\sim t_2\qedhere\]
\end{proof}

We hence have completeness for all dyadic fragments of quantum computation. By taking their union, we can get completeness for the ``whole dyadic fragment''.

\begin{defi}
Let $\cat{SOP}[\mathbb D] := \bigcup\limits_{k=1}^\infty \cat{SOP}[\frac1{2^{k}}]$ be the whole dyadic fragment of quantum computation.
\end{defi}

\begin{cor}[Completeness of $\cat{SOP}{[}\mathbb D{]}/\sim_{\raisebox{-1ex}{\hspace*{-2ex}\textnormal{\scriptsize TH'}}}$]
\[\forall t_1,t_2\in \cat{SOP}[{\textstyle\mathbb D}],~~\interp{t_1}=\interp{t_2}\iff t_1\underset{\textnormal{TH'}}\sim t_2\]
\end{cor}

\section{Summing and Concatenating $\cat{SOP}$-Morphisms}
\label{sec:sum-concat}

We show in this section two interesting constructions that allow us to perform operations on $\cat{SOP}$-morphisms that are not primitively doable in $\cat{SOP}$ or more generally in gate-based quantum computation, but necessary when considering Hamiltonian-based computation \cite{Shaikh2022sum}, namely their sum and concatenation. Specifically, for any two $\cat{SOP}$-morphisms $t_0,t_1:n\to m$, we want to be able to build terms $t_{\textit{sum}}$ and $t_{\textit{concat}}$ such that:
\begin{align*}
\interp{t_{\textit{sum}}} &= \interp{t_0}+\interp{t_1}\\
\interp{t_{\textit{concat}}} &= \ket0\otimes\interp{t_0} + \ket1\otimes\interp{t_1}
\end{align*}
As we will see, these constructions are well suited to the dyadic fragments, as they can be performed entirely inside them.

To do so, we need a notion of controlled $\cat{SOP}$-morphism, inferred from \cite{ZXNormalForm} and made systematic in the graphical framework in \cite{Jeandel2022Addition}. 
\begin{defi}
A $\cat{SOP}$-morphism $t:n+1\to m$ is called a \emph{controlled morphism} if: $$\interp{t\circ(\ket0\otimes id_n)} = \sum_{\vec y\in\{0,1\}^{n+m}} \ketbra{y_0,...,y_m}{y_{m+1},...,y_{m+n}} = \interp{H^n_m(1)}$$
where $H^n_m(1)$ is the H-spider from $\cat{ZH}$ with parameter $1$. It will be used as a shortcut notation to represent the linear map $\sum_{\vec y\in\{0,1\}^{n+m}} \ketbra{y_0,...,y_m}{y_{m+1},...,y_{m+n}}$ or equivalently the $\cat{SOP}$-morphism $\sum_{\vec y} \ketbra{y_0,...,y_m}{y_{m+1},...,y_{m+n}}$.\\
In a controlled morphism, the rightmost input is called the \emph{control input}. We also call the morphism $t\circ(\ket1\otimes id_n)$ the \emph{controlee}.
\end{defi}

\begin{exa}
\label{ex:controlled-identity}
$t:=\frac{1}{2} \sum e^{2i\pi\left(\frac{1}{2}y_{0}y_{1}y_{6} + \frac{1}{2}y_{1}y_{2}y_{6}\right)}\ket{y_{0}}\!\!\bra{y_{1}, y_{2}}
$ is a controlled morphism, as $\interp{t} = \bra0\otimes\interp{H^1_1(1)} + \bra1\otimes id_1$.
\end{exa}

In the following, we show how given two controlled morphisms, we can build a third controlled morphism whose controlee is the sum (resp.~the concatenation) of the controlees of the first two terms. Finally, we show how to turns any term $t$ into a controlled morphism whose controlee is precisely $t$. Since recovering the controlee out of a controlled term is easy (it suffices to apply $\ket1$ to the control input), the process towards building e.g.~the sum of two terms $t_0$ and $t_1$ would then be:
\begin{itemize}
\item build controlled terms $\Lambda t_0$ and $\Lambda t_1$ controlling $t_0$ and $t_1$ respectively
\item use the construction to build $\Lambda (t_0+t_1)$, a term controlling $t_0+t_1$
\item apply $\ket1$ to the control inputs to get $t_0+t_1$
\end{itemize}

\subsection{The Constructions}

Using this notion of controlled morphism, there exists a construction that will allow us to perform a sum of terms. In this construction, we need in particular:
$$t_+:=\frac{1}{2} \sum e^{2i\pi\left(\frac{1}{2}y_{0}y_{1}y_{3}\right)}\ket{y_{0}, y_{1}}\!\!\bra{y_{0}{\oplus}y_{1}}$$
and the family of morphisms:
$$\operatorname{cp}_n:n\to 2n:= \sum_{\vec y}\ketbra{\vec y,\vec y}{\vec y}$$
We also need the following identities:
\begin{lem}
$\interp{H_{m_1}^{n_1}(1)}\otimes\interp{H_{m_2}^{n_2}(1)} = \interp{H_{m_1+m_2}^{n_1+n_2}(1)}$
\end{lem}

\begin{proof}
This is obvious from the fact that:
\[\interp{H^n_m(1)} = \sum_{\vec y\in\{0,1\}^{n+m}} \ketbra{y_0,...,y_m}{y_{m+1},...,y_{m+n}}\qedhere\]
\end{proof}

\begin{lem}
$\quad~\left(\interp{H_0^n(1)}\otimes ~.~\right)\circ \interp{\operatorname{cp}_n} = (.)\qquad\qquad
(~.~\otimes \interp{H_0^n(1)})\circ \interp{\operatorname{cp}_n} = (.)$
\[\interp{\operatorname{cp}_m}^\dagger\circ\left(\interp{H_m^0(1)}\otimes ~.~\right) = (.)\qquad\qquad
\interp{\operatorname{cp}_m}^\dagger\circ\left(~.~\otimes \interp{H_m^0(1)}\right) = (.)\]
\end{lem}

\begin{proof}
The first equation is obtained by:
\begin{align*}
\left(\interp{H_0^n(1)}\otimes ~.~\right)\circ \interp{\operatorname{cp}_n} &= \left(\sum_{\vec y_1,\vec y_2}\ketbra{\vec y_2}{\vec y_1,\vec y_2}\right)\circ\left(\sum_{\vec y}\ketbra{\vec y,\vec y}{\vec y}\right)=\sum_{\vec y}\ketbra{\vec y}{\vec y} = id_n
\end{align*}
The other three equations can be obtained similarly.
\end{proof}

We may now build, from two controlled morphisms, a third controlled morphism whose controlee is the sum of the two first controlees:
\begin{prop}
\label{prop:sum-SOP}
Let $t_1,t_2:n+1\to m$ be two controlled morphisms. We define:
$$t:=\operatorname{cp}_m^\dagger\circ (t_1\otimes t_2) \circ (id_1\otimes\sigma_{1,n}\otimes id_n)\circ (t_+\otimes\operatorname{cp}_n)$$
Then:
$$\interp{t}=\bra0\otimes\interp{H_m^n(1)} + \bra1\otimes\left(\vphantom{\rule{1pt}{1em}}\interp{t_1\circ(\ket1\otimes id_n)}+\interp{t_2\circ(\ket1\otimes id_n)}\right)$$
\end{prop}

\begin{proof}
\phantomsection\label{prf:sum-SOP}
First, notice that $t_+\circ\ket0\rewrite{\text{HH}}\ket{0,0}$ and $t_+\circ\ket1\rewrite{\text{HH}}\sum\ket{y_0,y_0{\oplus}1}$.

Then:
\begin{align*}
\interp{t\circ(\ket0\otimes id_n)}
&= \interp{\operatorname{cp}_m^\dagger\circ (t_1\otimes t_2) \circ (id_1\otimes\sigma_{1,n}\otimes id_n)\circ ((t_+\circ\ket0)\otimes\operatorname{cp}_n)}\\
&= \interp{\operatorname{cp}_m^\dagger\circ (t_1\otimes t_2) \circ (id_1\otimes\sigma_{1,n}\otimes id_n)\circ (\ket{0,0}\otimes\operatorname{cp}_n)}\\
&= \interp{\operatorname{cp}_m}^\dagger\circ (\interp{t_1\circ(\ket0\otimes id_n)}\otimes \interp{t_2\circ(\ket0\otimes id_n)}) \circ \interp{\operatorname{cp}_n}\\
&= \interp{\operatorname{cp}_m}^\dagger\circ (\interp{H_m^n(1)}\otimes \interp{H_m^n(1)}) \circ \interp{\operatorname{cp}_n}
= \interp{H_m^n(1)}
\end{align*}
and:
\begin{align*}
\interp{t\circ(\ket1\otimes id_n)}
&= \interp{\operatorname{cp}_m^\dagger\circ (t_1\otimes t_2) \circ (id_1\otimes\sigma_{1,n}\otimes id_n)\circ ((t_+\circ\ket1)\otimes\operatorname{cp}_n)}\\
&= \interp{\operatorname{cp}_m^\dagger\circ (t_1\otimes t_2) \circ (id_1\otimes\sigma_{1,n}\otimes id_n)\circ \left(\left(\sum\ket{y_0,y_0{\oplus}1}\right)\otimes\operatorname{cp}_n\right)}\\
&= \interp{\operatorname{cp}_m^\dagger\circ (t_1\otimes t_2) \circ (id_1\otimes\sigma_{1,n}\otimes id_n)\circ \left(\ket{0,1}\otimes\operatorname{cp}_n\right)}\\
\tag*{$+ \interp{\operatorname{cp}_m^\dagger\circ (t_1\otimes t_2) \circ (id_1\otimes\sigma_{1,n}\otimes id_n)\circ \left(\ket{1,0}\otimes\operatorname{cp}_n\right)}$}\\
&= \interp{\operatorname{cp}_m}^\dagger\circ (\interp{t_1\circ(\ket0\otimes id_n)}\otimes \interp{t_2\circ(\ket1\otimes id_n)}) \circ \interp{\operatorname{cp}_n}\\
\tag*{$+ \interp{\operatorname{cp}_m}^\dagger\circ (\interp{t_1\circ(\ket1\otimes id_n)}\otimes \interp{t_2\circ(\ket0\otimes id_n)}) \circ \interp{\operatorname{cp}_n}$}\\
&= \interp{t_2\circ(\ket1\otimes id_n)}) + \interp{t_1\circ(\ket1\otimes id_n)}\qedhere
\end{align*}
\end{proof}

\begin{exa}
$t_1=\frac{1}{2} \sum e^{2i\pi\left(\frac{1}{2}y_{0}y_{1}y_{6} + \frac{1}{2}y_{1}y_{2}y_{6}\right)}\ket{y_{0}}\!\!\bra{y_{1}, y_{2}}
$ is a controlled morphism by Example \ref{ex:controlled-identity}. The morphism $t_2 = \sum \ketbra{y_0}{y_1,y_2}$ is also obviously a controlled morphism. $t_1$ controls $\begin{pmatrix}1&0\\0&1\end{pmatrix}$ and $t_2$ controls $\begin{pmatrix}1&1\\1&1\end{pmatrix}$. The above construction on $t_1$ and $t_2$ yields $$t\overset{\ast}{\rewrite{\text{HH}}} \frac{1}{4} \sum e^{2i\pi\left(\frac{1}{2}y_{2}y_{11}y_{6} + \frac{1}{2}y_{2}y_{3}y_{4} + \frac{1}{2}y_{0}y_{2}y_{4}\right)}\ket{y_{0}}\!\!\bra{y_{2}{\oplus}y_{6}, y_{3}}\quad\text{which controls }\begin{pmatrix}2&1\\1&2\end{pmatrix}.$$
To get the controlee, it then suffices to apply $\ket1$ on the first input of $t$:
$$t\circ(\ket1\otimes id) \overset{\ast}{\rewrite{\text{HH}}} \frac{1}{2} \sum e^{2i\pi\left(\frac{1}{2}y_{3}y_{4}y_{6} + \frac{1}{2}y_{0}y_{4} + \frac{1}{2}y_{3}y_{4} + \frac{1}{2}y_{0}y_{4}y_{6}\right)}\ket{y_{0}}\!\!\bra{y_{3}}$$
\end{exa}

The notion of controlled morphism also allows us to define a construction that will perform the concatenation of the controlees. Here again, we need a new building block:
$$t_c:=\sum \ket{y_{0}, y_{0}y_{1}{\oplus}y_{1}, y_{0}y_{1}}\!\!\bra{y_{1}}$$

\begin{prop}
\label{prop:concat-SOP}
Let $t_1,t_2:n+1\to m$ be two controlled morphisms. We define:
$$t:=(id_1\otimes\operatorname{cp}_m^\dagger)\circ (id_1\otimes t_1\otimes t_2) \circ (id_2\otimes\sigma_{1,n}\otimes id_n)\circ (t_c\otimes\operatorname{cp}_n)$$
Then:
$$\interp{t}=\bra0\otimes\interp{H_{m+1}^n(1)} + \bra1\otimes\left(\vphantom{\rule{1pt}{1em}}\ket0\otimes\interp{t_1\circ(\ket1\otimes id_n)}+\ket1 \otimes\interp{t_2\circ(\ket1\otimes id_n)}\right)$$
i.e.~$t$ controls the concatenation of $t_1$ and $t_2$'s controlees.
\end{prop}

\begin{proof}
\phantomsection\label{prf:concat-SOP}
Notice first that $t_c\circ \ket0 \rewrite{\text{HH}} \sum \ket{y_0,0,0}$, $(\bra0\otimes id_2)\circ t_c\ket1 \rewrite{\text{HH}} \ket{1,0}$ and $(\bra1\otimes id_2)\circ t_c\ket1 \rewrite{\text{HH}} \ket{0,1}$.

Let $t$ be defined as above. First, we have:
\begin{align*}
\interp{t\circ(\ket0\otimes id_n)}\hspace*{-3em}&\\
&=\interp{(id_1\otimes\operatorname{cp}_m^\dagger)\circ (id_1\otimes t_1\otimes t_2) \circ (id_2\otimes\sigma_{1,n}\otimes id_n)\circ \left(\left(\sum \ket{y_0,0,0}\right)\otimes\operatorname{cp}_n\right)}\\
&= \interp{\left(\sum \ket{y_0}\right)\otimes\left(\operatorname{cp}_m^\dagger\circ \left((t_1 \circ(\ket0\otimes id_n))\otimes (t_2\circ(\ket0\otimes id_n))\right) \circ \operatorname{cp}_n\right)}\\
&= \interp{H_1^0(1)}\otimes\left(\interp{\operatorname{cp}_m}^\dagger\circ \left(\interp{t_1 \circ(\ket0\otimes id_n)}\otimes \interp{t_2\circ(\ket0\otimes id_n)}\right) \circ \interp{\operatorname{cp}_n}\right)\\
&= \interp{H_1^0(1)}\otimes\left(\interp{\operatorname{cp}_m}^\dagger\circ \left(\interp{H_m^n(1)}\otimes \interp{H_m^n(1)}\right) \circ \interp{\operatorname{cp}_n}\right)
= \interp{H_1^0(1)}\otimes\interp{H_m^n(1)}\\
&= \interp{H_{m+1}^n(1)}
\end{align*}
Then:
\begin{align*}
\interp{(\bra0\otimes id_m)\circ t\circ(\ket0\otimes id_n)}\hspace*{-5em}&\\
&= \interp{\operatorname{cp}_m^\dagger\circ (t_1\otimes t_2) \circ (id_1\otimes\sigma_{1,n}\otimes id_n)\circ ((\bra0\otimes id_2)\circ t_c\ket1)\otimes\operatorname{cp}_n)}\\
&= \interp{\operatorname{cp}_m^\dagger\circ (t_1\otimes t_2) \circ (id_1\otimes\sigma_{1,n}\otimes id_n)\circ (\ket{1,0}\otimes\operatorname{cp}_n)}\\
&= \interp{\operatorname{cp}_m}^\dagger\circ (\interp{t_1 \circ(\ket1\otimes id_n)}\otimes \interp{t_2 \circ(\ket0\otimes id_n)}) \circ \interp{\operatorname{cp}_n}\\
&= \interp{\operatorname{cp}_m}^\dagger\circ (\interp{t_1 \circ(\ket1\otimes id_n)}\otimes \interp{H_m^n(1)}) \circ \interp{\operatorname{cp}_n}\\
&= \interp{t_1 \circ(\ket1\otimes id_n)}
\end{align*}
and similarly:
\begin{align*}
\interp{(\bra1\otimes id_m)\circ t\circ(\ket0\otimes id_n)}
= \interp{t_2 \circ(\ket1\otimes id_n)}&\qedhere
\end{align*}
\end{proof}

\subsection{Controlling $\cat{SOP}$-Morphisms}

Critical to the previous two constructions is the existence, for every morphism $t$ of another morphism that controls it.

First, we need to be able, in all generality, to control arbitrary complex scalars.

\begin{defi}
Let $s\in\mathbb C$. If $s\neq0$, let $n:=\lceil \log_2(|s|+1)\rceil$, $\alpha:=\arccos\left(\frac{|s|}{2^n-1}\right)$ and $\theta = \arg(s)$. Then $s = (2^n-1)\cos\alpha e^{i\theta}$. We hence define:
$$\mathrm\Lambda' s := \begin{cases}
\bra0 &\text{ if } s=0\\
\frac14\sum e^{2i\pi \left(\frac{y_1...y_ny''}2 + \frac{y''(1+y_0)}2 + \frac{2\alpha}{2\pi}y_0y' + \frac{\theta-\alpha}{2\pi}y_0\right)}\bra{y_0}&\text{ if }s\neq0
\end{cases}$$
\end{defi}

\begin{prop}
\label{prop:control-scalar'}
For any $s\in\mathbb C$, 
$\interp{\mathrm\Lambda' s} = \bra0+s\bra1$.
\end{prop}

\begin{proof}
\phantomsection\label{prf:control-scalar'}
If $s=0$, the result is obvious. Otherwise, notice that $\sum\limits_{\vec y\in\{0,1\}^n} e^{2i\pi \frac{y_1...y_n}2} = 2^n-2$. Then:
\begin{align*}
\interp{\mathrm\Lambda' s\circ\ket0}
&= \interp{\frac14\sum_{\vec y\setminus\{y_0\}} e^{2i\pi \left(\frac{y_1...y_ny''}2 + \frac{y''}2\right)}}
= \interp{\frac12\sum_{\vec y\setminus\{y_0,y'\}} e^{2i\pi \left(\frac{y_1...y_ny''}2 + \frac{y''}2\right)}}\\
&= \interp{\frac12\sum_{\vec y\setminus\{y_0,y',y''\}} e^{2i\pi \left(0\right)}} - \interp{\frac12\sum_{\vec y\setminus\{y_0,y',y''\}} e^{2i\pi \left(\frac{y_1...y_n}2\right)}} = \frac12(2^n -(2^n-2)) = 1
\end{align*}
\begin{align*}
\interp{\mathrm\Lambda' s\circ\ket1}
&= \interp{\frac14\sum_{\vec y\setminus\{y_0\}} e^{2i\pi \left(\frac{y_1...y_ny''}2 + \frac{2\alpha}{2\pi}y' + \frac{\theta-\alpha}{2\pi}\right)}}\\
&= \frac14\interp{\hspace*{-2em}\sum_{\qquad\vec y\setminus\{y_0,y'\}} \hspace*{-2em}e^{2i\pi \left(\frac{y_1...y_ny''}2\right)}} \interp{\sum_{y'} e^{2i\pi\left(\frac{2\alpha}{2\pi}y' + \frac{\theta-\alpha}{2\pi}\right)}}\\
&= \frac14(2^{n+1}-2)e^{i(\theta-\alpha)}(1+e^{2i\alpha})=(2^n-1)\cos\alpha e^{i\theta} = s\qedhere
\end{align*}
\end{proof}

This construction is interesting as it requires a fairly small number of variables. However, scalars that we can find in the dyadic fragments ($\cat{SOP}[\frac1{2^n}]$) may require their control to be outside them. For instance, in $\mathrm\Lambda'\frac12$, we need $\alpha=\frac\pi3$, which is not a dyadic multiple of $\pi$. We hence give another definition for the controlled scalar:

\begin{defi}
Let $s\in\mathbb C$. If $s\neq0$, let $n:=\lceil \log_2(|s|)\rceil$, $\alpha:=\arccos\left(\frac{|s|}{2^n}\right)$ and $\theta = \arg(s)$. Then $s = 2^n\cos\alpha e^{i\theta}$. We define:
$$\mathrm\Lambda s := \begin{cases}
\bra0&\text{ if }s=0\\
\frac1{2^{n+1}}\sum e^{2i\pi \left(\sum\limits_{i=1}^n \frac{y_iy_i'(1+y_0)}2 + \frac{2\alpha}{2\pi}y_0y' + \frac{\theta-\alpha}{2\pi}y_0\right)}\bra{y_0}&\text{if }n\geq0\\
2^{2n-1}\sum e^{2i\pi \left(\sum\limits_{i=1}^{|n|} \frac{y_iy_i'y_0}2 + \frac{2\alpha}{2\pi}y_0y' + \frac{\theta-\alpha}{2\pi}y_0\right)}\bra{y_0}&\text{if }n<0
\end{cases}$$
\end{defi}

\begin{prop}
\label{prop:control-scalar}
For any $s\in\mathbb C$, 
$\interp{\mathrm\Lambda s} = \bra0+s\bra1$.
\end{prop}

\begin{proof}
\phantomsection\label{prf:control-scalar}
If $s=0$, the result is obvious.

Then, if $n\geq0$:
\begin{align*}
\mathrm\Lambda s\circ\ket0
\rewrite{\text{HH}}
\frac1{2^{n+1}}\sum e^{2i\pi \left(\sum\limits_{i=1}^n \frac{y_iy_i'}2 \right)}
\rewrite{\text{Elim}}
\frac1{2^n}\sum e^{2i\pi \left(\sum\limits_{i=1}^n \frac{y_iy_i'}2 \right)}
\rewrite{\text{HH}^n}
1
\end{align*}
and:
\begin{align*}
\mathrm\Lambda s\circ\ket1
\rewrite{\text{HH}}
\frac1{2^{n+1}}\sum e^{2i\pi \left(\sum\limits_{i=1}^n 0 +\frac{2\alpha}{2\pi}y' + \frac{\theta-\alpha}{2\pi} \right)}
\rewrite{\text{Elim}^{2n}}
2^{n-1}\sum e^{2i\pi \left(\frac{2\alpha}{2\pi}y' + \frac{\theta-\alpha}{2\pi} \right)}
\end{align*}
So $\interp{\mathrm\Lambda s\circ\ket1} = 2^{n-1}e^{i(\theta-\alpha)}(1+e^{2i\alpha}) = 2^n\cos\alpha e^{i\theta}=s$.

Now, if $n<0$:
\begin{align*}
\mathrm\Lambda s\circ\ket0
\rewrite{\text{HH}}
2^{2n-1}\sum e^{2i\pi \left(0\right)}
\rewrite{\text{Elim}^{2|n|+1}} 1
\end{align*}
and
\begin{align*}
\mathrm\Lambda s\circ\ket1
\rewrite{\text{HH}}
2^{2n-1}\sum e^{2i\pi \left(\sum\limits_{i=1}^{|n|} \frac{y_iy_i'}2 + \frac{2\alpha}{2\pi}y' + \frac{\theta-\alpha}{2\pi}\right)}
\rewrite{\text{HH}^{|n|}}
2^{n-1}\sum e^{2i\pi \left(\frac{2\alpha}{2\pi}y' + \frac{\theta-\alpha}{2\pi}\right)}
\end{align*}
so again $\interp{\mathrm\Lambda s\circ\ket1}=s$.
\end{proof}

This time, if $s = \frac1{\sqrt2^p} = 2^{\lceil\frac p2\rceil}\cos{\frac\pi4(p\bmod2)}$, controlling $s$ gives a morphism in a dyadic fragment of $\cat{SOP}$.

We now give a general construction for controlling an arbitrary $\cat{SOP}$-term.

\begin{defi}
Let $t = s\sum\limits_{\vec y} e^{2i\pi P} \ketbra{\vec O}{\vec I} \in\cat{SOP}$. We define:
$$\overset{\sim}{\mathrm\Lambda}t := 
\frac1{2^{|\vec y|}}\sum_{\substack{\vec y, x_0,\\\vec x_1,\vec x_2}} e^{2i\pi x_0P} \ketbra{x_0\vec O{\oplus}\vec x_1{\oplus}x_0\vec x_1}{x_0,x_0\vec I{\oplus}\vec x_2{\oplus}x_0\vec x_2}$$
where
\[x_0\vec O{\oplus}\vec x_1{\oplus}x_0\vec x_1:=(x_0 O_1{\oplus} x_{11}{\oplus}x_0 x_{11},...,x_0 O_m{\oplus}x_{1m}{\oplus}x_0x_{1m}),\]
and similarly
\[x_0\vec I{\oplus}\vec x_2{\oplus}x_0\vec x_2:=(x_0 I_1{\oplus} x_{21}{\oplus}x_0 x_{21},...,x_0 I_{m+n}{\oplus}x_{2n}{\oplus}x_0x_{2n}).\]
We finally define:
\[\mathrm\Lambda t := \left(\mathrm\Lambda (s{2^{|\vec y|-n-m}}) \otimes \overset{\sim}{\mathrm\Lambda}t\right) \circ (\operatorname{cp}_1\otimes id_n)\]
\end{defi}

\begin{prop}
\label{prop:control-term}
$\mathrm\Lambda t$ is a controlled morphism, that controls $t$, i.e.:
$$\interp{\mathrm\Lambda t} = \bra0\otimes\interp{H_m^n(1)} + \bra 1 \otimes \interp{t}$$
\end{prop}

\begin{proof}
\phantomsection\label{prf:control-term}
We have
\begin{align*}
\interp{\mathrm\Lambda t }\circ (\ket0\otimes id_n)
&= \interp{\mathrm\Lambda (s{2^{|\vec y|-n-m}})\circ\ket0} \otimes \interp{\overset{\sim}{\mathrm\Lambda}t\circ (\ket0\otimes id_n)}
= 1 \otimes \frac1{2^{|\vec y|}}\sum_{\vec y,\vec x1,\vec x_2} e^{2i\pi 0} \ketbra{\vec x_1}{\vec x_2}\\
&= \sum_{\vec x1,\vec x_2} \ketbra{\vec x_1}{\vec x_2} = \interp{H_m^n(1)}
\end{align*}
and
\begin{align*}
\interp{\mathrm\Lambda t }\circ (\ket1\otimes id_n)
&= \interp{\mathrm\Lambda (s{2^{|\vec y|-n-m}})\circ\ket1} \otimes \interp{\overset{\sim}{\mathrm\Lambda}t\circ (\ket1\otimes id_n)}\\
&= \frac{s{2^{|\vec y|}}}{2^{n+m}} \otimes \frac1{2^{|\vec y|}}\sum_{\vec y,\vec x1,\vec x_2} e^{2i\pi P} \ketbra{\vec O}{\vec I}\\
&= \frac{s}{2^{n+m}}\sum_{\vec y,\vec x1,\vec x_2} e^{2i\pi P} \ketbra{\vec O}{\vec I} = s\sum_{\vec y} e^{2i\pi P} \ketbra{\vec O}{\vec I} = t
\qedhere
\end{align*}
\end{proof}

Notice that if $t\in\cat{SOP}[\frac1{2^n}]$, then $\mathrm\Lambda t \in \cat{SOP}[\frac1{2^{\max(3,n)}}]$. Notice also that, for $t:n\to m$, only $n+m+1$ additional variables are needed in $\overset{\sim}{\mathrm\Lambda}t$. However, the controlled scalar will require around $2(|\vec y|-n-m+\log_2(|s|))$ variables, and the phase polynomial gets one degree higher. Hence, it may be useful to put additional information to use if we are provided some.

\begin{defi}
Let $t = s\sum\limits_{\vec y} e^{2i\pi P} \ketbra{\vec O}{\vec I} \in\cat{SOP}$. For any $\vec v_1\in\{0,1\}^m,\vec v_2\in\{0,1\}^n$ and $\rho e^{i\theta}\neq 0$, we define:
\begin{align*}
\overset{\sim}{\mathrm\Lambda}_2(\,t\,|\vec v_1&,\vec v_2, \rho,\theta) :=\\
&\frac1{\rho2^{n+m}}\sum e^{2i\pi \left(-\frac\theta{2\pi}+ P + \frac12\left(\vec{\widehat O} + \vec x_1 + x_0\vec x_1 + x_0\vec v_1\right)\cdot\vec x_1' + \frac12\left(\vec{\widehat I} + \vec x_2 + x_0\vec x_2 + x_0\vec v_2\right)\cdot\vec x_2'\right)}\ketbra{\vec x_1}{x_0{\oplus}1,\vec x_2}
\end{align*}
and:
\[{\mathrm\Lambda}_2(\,t\,|\vec v_1,\vec v_2, \rho,\theta) :=\left(\mathrm\Lambda (s\rho e^{i\theta}) \otimes \overset{\sim}{\mathrm\Lambda}_2(\,t\,|\vec v_1,\vec v_2, \rho,\theta)\right) \circ (\operatorname{cp}_1\otimes id_n)\]
\end{defi}

\begin{prop}
\label{prop:control-term-2}
For any $\vec v_1\in\{0,1\}^m,\vec v_2\in\{0,1\}^n$ and $\rho e^{i\theta}\neq 0$ such that $\bra{\vec v_1}\interp{t}\ket{\vec v_2} = \rho e^{i\theta}$, 
${\mathrm\Lambda}_2(\,t\,|\vec v_1,\vec v_2, \rho, \theta)$ is a controlled morphism, that controls $t$, i.e.:
$$\interp{{\mathrm\Lambda}_2(\,t\,|\vec v_1,\vec v_2, \rho, \theta)} = \bra0\otimes\interp{H_m^n(1)} + \bra 1 \otimes \interp{t}$$
\end{prop}

\begin{proof}
\phantomsection\label{prf:control-term-2}
Let $t = s\sum\limits_{\vec y} e^{2i\pi P} \ketbra{\vec O}{\vec I} \in\cat{SOP}$, $\vec v_1\in\{0,1\}^m,\vec v_2\in\{0,1\}^n$ and $\rho e^{i\theta}\neq 0$ such that $\bra{\vec v_1}\interp{t}\ket{\vec v_2} = \rho e^{i\theta}$. We have:
\begin{align*}
\interp{\overset{\sim}{\mathrm\Lambda}_2(\,t\,|\vec v_1,\vec v_2, \rho, \theta)\circ(\ket0\otimes id_n)}\hspace*{-7em}&\\
&= \interp{\frac1{\rho2^{n+m}}\sum_{\substack{\vec y,\vec x_1,\vec x_2\\\vec x_1',\vec x_2'}} e^{2i\pi \left(-\frac\theta{2\pi}+ P + \frac12\left(\vec{\widehat O} + \vec v_1\right)\cdot\vec x_1' + \frac12\left(\vec{\widehat I} + \vec v_2\right)\cdot\vec x_2'\right)}\ketbra{\vec x_1}{\vec x_2}}\\
&=\frac1{\rho e^{i\theta}}\interp{\frac1{2^{n+m}}\sum_{\substack{\vec y,\vec x_1',\vec x_2'}} e^{2i\pi \left(P + \frac12\left(\vec{\widehat O} + \vec v_1\right)\cdot\vec x_1' + \frac12\left(\vec{\widehat I} + \vec v_2\right)\cdot\vec x_2'\right)}}
\interp{\sum_{\vec x_1,\vec x_2} \ketbra{\vec x_1}{\vec x_2}}\\
&= \frac1{\rho e^{i\theta}} \interp{\bra{\vec v_1}t\ket{\vec v_2}}\interp{H_m^n(1)} = \interp{H_m^n(1)}
\end{align*}
so
\begin{align*}
\interp{{\mathrm\Lambda}_2(\,t\,|\vec v_1,\vec v_2, \rho,\theta)}\circ(\ket0\otimes id_n)
&= \interp{\mathrm\Lambda (s\rho e^{i\theta})\circ\ket0} \otimes \interp{\overset{\sim}{\mathrm\Lambda}_2(\,t\,|\vec v_1,\vec v_2, \rho, \theta)\circ(\ket0\otimes id_n)}\\
&=\interp{H_m^n(1)}
\end{align*}
And:
\begin{align*}
\overset{\sim}{\mathrm\Lambda}_2(\,t\,|\vec v_1,\vec v_2, \rho,\theta)\circ(\ket0\otimes id_n)
&\rewrite{\text{HH}}
\frac1{\rho2^{n+m}}\sum e^{2i\pi \left(-\frac\theta{2\pi}+ P + \frac12\left(\vec{\widehat O} + \vec x_1 \right)\cdot\vec x_1' + \frac12\left(\vec{\widehat I} + \vec x_2 \right)\cdot\vec x_2'\right)}\ketbra{\vec x_1}{\vec x_2}\\
&\rewrite{\text{HH}^{n+m}}
\frac1{\rho}\sum e^{2i\pi\left(-\frac\theta{2\pi} + P\right)}\ketbra{\vec O}{\vec I}
\end{align*}
so
\begin{align*}
\interp{{\mathrm\Lambda}_2(\,t\,|\vec v_1,\vec v_2, \rho,\theta)}\circ(\ket1\otimes id_n)
&= \interp{\mathrm\Lambda (s\rho e^{i\theta})\circ\ket1} \otimes \interp{\overset{\sim}{\mathrm\Lambda}_2(\,t\,|\vec v_1,\vec v_2, \rho,\theta)\circ(\ket1\otimes id_n)}\\
&= s\rho e^{i\theta}\interp{\frac1{\rho}\sum e^{2i\pi\left(-\frac\theta{2\pi} + P\right)}\ketbra{\vec O}{\vec I}}\\
&= \interp{s\sum e^{2i\pi P}\ketbra{\vec O}{\vec I}} = \interp{t}
\end{align*}
This proves $\interp{{\mathrm\Lambda}_2(\,t\,|\vec v_1,\vec v_2, \rho, \theta)} = \bra0\otimes\interp{H_m^n(1)} + \bra 1 \otimes \interp{t}$.
\end{proof}

In order to control $t:n\to m$ in this version, we only need around $2(n+m+\log_2(|s\rho e^{i\theta}|))$ additional variables, and the initial phase polynomial is simply added another polynomial. However, it requires some prior knowledge on the linear map that is represented by $t$, namely, one of its coefficients. Another caveat of importance is that $\frac1{\rho e^{i\theta}}$, which is needed in the definition of the control, will in general get us out of the dyadic fragments.

\begin{rem}
In the definitions of $\mathrm\Lambda$ and $\mathrm\Lambda_2$ for $\cat{SOP}$-terms, we used the control of scalars denoted $\mathrm\Lambda$. However, any other notion of control of scalars (like $\mathrm\Lambda'$ for instance) would work.
\end{rem}

\section{Conclusion and Discussion}

We have given a new rewrite system for the (balanced) Toffoli-Hadamard fragment of Sums-Over-Paths, and showed the induced equational theory to be complete. We then extended this rewrite strategy by adding a single new rewrite, which we then proved to be complete for any dyadic fragment of quantum computation. As expected from the universality of the fragments at hand, we do not get all the nice properties of the rewriting in the Clifford fragment. In particular, we showed that the rewrite strategies given above are not confluent, and that the size of the terms may grow exponentially when rules are applied carelessly. Whether one of the above two drawbacks can be removed by a different rewrite system remains an open question. In particular, we have exhibited two rules that, in order to be applied properly, require solving a hard but well studied problem, i.e.~that of building Groebner bases. We wonder if this is a sufficient price to pay to get confluence. If not, it would be interesting to see if  a completion of the rewrite strategy à la Knuth-Bendix is possible in this framework, and if such completion would eventually yield a finite set of rules (note that this would not necessarily make problems like circuit equivalence classically tractable, as the size of the phase polynomial may still grow exponentially quickly, and application of some rules may require solving hard problems).

On a more foundational note, we conjecture that the rewrite strategy $\underset{\text{TH}}\longrightarrow$ is minimal, i.e.~that none of the rewrites can be derived from the others. The same question can be asked in $\cat{ZH}$, but no proof (confirming or refuting it) has been provided yet. The graphical study of the rewrite rules of \autoref{sec:rewrites-in-ZH} hint at directions to simplify the rules of $\cat{ZH}$ thanks to their interpretation in $\cat{SOP}$, or to transport potential proofs of minimality between the two. The completeness result derived from that of the $\cat{ZH}$, and this small study of how the rewrites translate as ZH transformations, really show how the two formalisms give different and complementary approaches to rewriting and simplifying representations of quantum processes.

Thanks to the proximity of $\cat{SOP}$ with the $\cat{ZH}$-calculus and to the theory developed in the graphical setting to perform sums and concatenations of terms, we were able to transport these constructions in the $\cat{SOP}$ formalism. This proves very useful when analysing Hamiltonian-based quantum computation, which heavily relies on building large Hamiltonians by sums of smaller terms.

All the work presented here was in the case of \emph{balanced} sum-over-paths. Adding unbalancedness as in \cite{amy2023complete} allows for reductions of equivalent terms with fewer variables, at the cost of a larger space in which the morphisms live. While having fewer variables is crucial when performing simulation (which implies actually computing the term), the extra degree of freedom that amplitude polynomials offer also makes issues like confluence potentially harder to tackle. We wonder if a proper unbalanced Sum-Over-Paths formalism can be stated for the Toffoli-Hadamard fragment and the dyadic fragments.

\bibliography{bibli}

\end{document}